\documentclass[oneside,english]{amsart}
\usepackage[T1]{fontenc}
\usepackage[latin9]{inputenc}
\usepackage{textcomp}
\usepackage{amstext}
\usepackage{amsthm}
\usepackage{esint}

\makeatletter
\numberwithin{equation}{section}
\numberwithin{figure}{section}
\theoremstyle{plain}
\newtheorem{thm}{\protect\theoremname}[section]
  \theoremstyle{remark}
  \newtheorem{rem}[thm]{\protect\remarkname}
  \theoremstyle{plain}
  \newtheorem{prop}[thm]{\protect\propositionname}
  \theoremstyle{plain}
  \newtheorem{lem}[thm]{\protect\lemmaname}
  \theoremstyle{definition}
  \newtheorem{example}[thm]{\protect\examplename}
  \theoremstyle{plain}
  \newtheorem{cor}[thm]{\protect\corollaryname}

\def\makebbb#1{
    \expandafter\gdef\csname#1\endcsname{
        \ensuremath{\Bbb{#1}}}
}\makebbb{R}\makebbb{N}\makebbb{Z}\makebbb{C}\makebbb{H}\makebbb{E}\makebbb{H}\makebbb{P}\makebbb{B}\makebbb{Q}\makebbb{E}

\usepackage{babel}

\usepackage{babel}

\makeatother

\usepackage{babel}

\renewcommand{\vec}[1]{\mathbf{#1}}

\makeatother

\usepackage{babel}
  \providecommand{\corollaryname}{Corollary}
  \providecommand{\examplename}{Example}
  \providecommand{\lemmaname}{Lemma}
  \providecommand{\propositionname}{Proposition}
  \providecommand{\remarkname}{Remark}
\providecommand{\theoremname}{Theorem}

\begin{document}

\title{Propagation of chaos, Wasserstein gradient flows and toric Kähler-Einstein
metrics}

\author{Robert J. Berman, Magnus Önnheim}

\email{robertb@chalmers.se, onnheimm@chalmers.se}

\address{Department of Mathematical Sciences, Chalmers University of Technology
and University of Gothenburg, 412 96 Göteborg, Sweden}
\begin{abstract}
Motivated by a probabilistic approach to Kähler-Einstein metrics we
consider a general non-equilibrium statistical mechanics model in
Euclidean space consisting of the stochastic gradient flow of a given
(possibly singular) quasi-convex N-particle interaction energy. We
show that a deterministic ``macroscopic'' evolution equation emerges
in the large N-limit of many particles. This is a strengthening of
previous results which required a uniform two-sided bound on the Hessian
of the interaction energy. The proof uses the theory of weak gradient
flows on the Wasserstein space. Applied to the setting of permanental
point processes at ``negative temperature'' the corresponding limiting
evolution equation yields a drift-diffusion equation, coupled to the
Monge-Ampère operator, whose static solutions correspond to toric
Kähler-Einstein metrics. This drift-diffusion equation is the gradient
flow on the Wasserstein space of probability measures of the K-energy
functional in Kähler geometry and it can be seen as a fully non-linear
version of various extensively studied dissipative evolution equations
and conservations laws, including the Keller-Segel equation and Burger's
equation. We also obtain a real probabilistic (and tropical) analog
of the complex geometric Yau-Tian-Donaldson conjecture in this setting.
In a companion paper applications to singular pair interactions  are
given.

\tableofcontents{}
\end{abstract}

\maketitle

\section{Introduction}

The present work is motivated by the probabilistic approach to the
construction of canonical metrics, or more precisely Kähler-Einstein
metrics, on complex algebraic varieties introduced in \cite{berm5,berm2},
formulated in terms of certain $\beta-$deformations of determinantal
(fermionic) point processes. The approach in \cite{berm5,berm2} uses
ideas from equilibrium statistical mechanics (Boltzmann-Gibbs measures)
and the main challenge concerns the existence problem for Kähler-Einstein
metrics on a complex manifold $X$ with \emph{positive} Ricci curvature,
which is closely related to the seminal Yau-Tian-Donaldson conjecture
in complex geometry. In this paper, which is one in a series, we will
be concerned with a dynamic version of the probabilistic approach
in \cite{berm5,berm2}. In other words, we are in the realm of non-equilibrium
statistical mechanics, where the relaxation to equilibrium is studied.
As the general complex geometric setting appears to be extremely challenging,
due to the severe singularities and non-linearity of the corresponding
interaction energies, we will here focus on the real analog of the
complex setting introduced in \cite{b}, taking place in $\R^{n}$
and which corresponds to the case when $X$ is a\emph{ toric} complex
algebraic variety. As explained in \cite{b} in this real setting
the determinantal (fermionic) processes are replaced by\emph{ permanental}
(bosonic) processes and convexity plays the role of positive Ricci
curvature/ plurisubharmonicity (see Section \ref{sub:The-complex-geometric}
for some geometric background). 

Our main result (Theorem \ref{thm:dynamic intro}) shows that a deterministic
evolution equation on the space of all probability measures on $\R^{n}$
emerges from the underlying stochastic dynamics, which as explained
below can be seen as a new ``propagation of chaos'' result. The
evolution equation in question is a drift-diffusion equation coupled
to the fully non-linear real Monge-Ampère operator. It exhibits a
phase transition at a certain geometrically determined critical parameter.
It turns out that in the case of the real line (i.e. $n=1)$ this
equation is closely related to various extensively studied evolution
equations, notably the Keller-Segel equation in chemotaxis \cite{k-s},
Burger's equation \cite{ho,f-b} in the theory of non-linear waves
and scalar conservation laws and the deterministic version of the
Kardar\textendash Parisi\textendash Zhang (KPZ) equation describing
surface growth \cite{kpz}. In the higher dimensional real case the
equation can be viewed as a dissipative viscous version of the semi-geostrophic
equation appearing in dynamic meteorology (see \cite{l,a-c-d-f} and
references therein). Moreover, closely related evolution equations
appear in cosmology and in particular in Brenier's approach to the
Zeldovich model used in the early universe reconstruction problem
\cite{s-z,f-m-m-s,Br2,br3}. For the corresponding \emph{static} problem
we establish a real analog of the Yau-Tian-Donaldson conjecture (Theorem
\ref{thm:static toric}) which, in particular, yields a probabilistic
construction of toric Kähler-Einstein metrics (see Section \ref{sub:The-complex-geometric}
for the relations to complex geometry).

As we were not able to deduce the type of propagation of chaos result
we needed from previous general results and approaches the main body
of the paper establishes the appropriate propagation of chaos result,
which, to the best of our knowledge, is new and hopefully the result,
as well as the method of proof, is of independent interest. As will
be clear below our approach heavily relies on the theory of weak gradient
flows on the Wasserstein $L^{2}-$space $\mathcal{P}_{2}(\R^{n})$
of probability measure on $\R^{n}$ developed in the seminal work
of Ambrosio-Gigli-Savare \cite{a-g-s}, which provides a rigorous
framework for the Otto calculus \cite{ot}. In particular, as in \cite{a-g-s}
convexity (or more generally $\lambda-$convexity) plays a prominent
role. Our limiting evolution equation will appear as the gradient
flow on $\mathcal{P}_{2}(\R^{n})$ of a certain free energy type functional
$F.$ Interestingly, as observed in \cite{berm4} the functional $F$
may be identified with Mabuchi's K-energy functional on the space
of Kähler metrics, which plays a key role in Kähler geometry and whose
gradient flow with respect to different metrics (the Mabuchi-Donaldson-Semmes
metric and Calabi's gradient metric) are the renowned\emph{ }Calabi
flow and Kähler-Ricci flow, respectively \cite{c-z}. The regularity
and large time properties of the evolution equation appearing here
will be studied elsewhere \cite{b-l,b-o}. 

In the remaining part of the introduction we will state our main results:
first, a general propagation of chaos result assuming a uniform Lipschitz
and convexity assumption on the interaction energy and then the application
to permanental point processes and toric Kähler-Einstein metrics.
In the companion paper \cite{b-o0} we will give a more general formulation
of the propagation of chaos result, by relaxing some of the assumptions
(in particular, this will yield sharp convergence results for strongly
singular repulsive pair interactions when $n=1).$

\subsection{Propagation of chaos and Wasserstein gradient flows}

Consider a system of $N$ identical particles diffusing on the $n-$dimensional
Euclidean space $X:=\R^{n}$ and interacting by a symmetric energy
function $E^{(N)}(x_{1},x_{2},....,x_{N})$. At a fixed inverse temperature
$\beta$ the distribution of particles at time $t$ is, according
to non-equilibrium statistical mechanics, described by the following
system of stochastic differential equations (SDEs), under suitable
regularity assumptions on $E^{(N)}:$

\begin{equation}
dx_{i}(t)=-\frac{\partial}{\partial x_{i}}E^{(N)}(x_{1},x_{2},....,x_{N})dt+\sqrt{\frac{2}{\beta}}dB_{i}(t),\label{eq:sde general intro}
\end{equation}
where $B_{i}$ denotes $N$ independent Brownian motions on $\R^{n};$
the equation is called the (overdamped) \emph{Langevin equation} in
the physics literature. In other words, this is the Ito diffusion
on $\R^{n}$ describing the (downward) gradient flow of the function
$E^{(N)}$ on the configuration space $X^{N}$ perturbed by a noise
term. A classical problem in mathematical physics going back to Boltzmann
and made precise by Kac \cite{kac} is to show that, in the many particle
limit where $N\rightarrow\infty,$ a\emph{ deterministic }macroscopic
evolution emerges from the stochastic microscopic dynamics described
by \ref{eq:sde general intro}. More precisely, denoting by $\delta_{N}$
the empirical measures 
\begin{equation}
\delta_{N}:=\frac{1}{N}\sum\delta_{x_{i}},\label{eq:empiric measure intro}
\end{equation}
 the SDEs \ref{eq:sde general intro} define a curve $\delta_{N}(t)$
of random measures on $X.$ The problem is to show that, if at the
initial time $t=0$ the random variables $x_{i}$ are independent
with identical distribution $\mu_{0}$ then at any later time $t$
the empirical measure $\delta_{N}(t)$ converges in law to a deterministic
curve $\mu_{t}$ of measures on $\R^{n}$ 
\begin{equation}
\lim_{N\rightarrow\infty}\delta_{N}(t)=\mu_{t}\label{eq:convergence empirical measure intro}
\end{equation}
 In the terminology of Kac \cite{kac} (see also \cite{sn}) this
means that \emph{propagation of chaos} holds at any time $t.$ It
should be stressed that the previous statement admits a pure PDE formulation,
not involving any stochastic calculus (see Section \ref{sub:The-forward-Kolmogorov})
and it is this analytic point of view that we will adopt here. \footnote{From a differential geometric point of view the SDEs \ref{eq:sde general intro}
correspond, under the transformation $\rho\mapsto e^{E/2}\rho,$ to
the heat flow on $X^{N}$ of the Witten Laplacian of the ``Morse
function'' $E,$ but we will not elaborate  on this point here. }

Of course, if propagation of chaos is to hold then some consistency
assumptions have to be made on the sequence $E^{(N)}$ of energy functions
as $N$ tends to infinity. The standard assumption in the literature
ensuring that propagation of chaos does hold is that $E^{(N)}(x_{1},x_{2},....,x_{N})$
can be as written as 
\begin{equation}
E^{(N)}(x_{1},x_{2},....,x_{N})=NE(\delta_{N})\label{eq:energy as funtion of empir measure intro}
\end{equation}
 for a fixed functional $E$ on the space of $\mathcal{P}(X)$ of
all probability measures on $X,$ where $E$ is assumed to have appropriate
regularity properties (to be detailed below). This is sometimes called
a mean field model. By the results in \cite{b-h,sn,da-g,mmw} it then
follows that the limit $\mu_{t}(=\rho_{t}dx)$ with initial data $\mu_{0}(=\rho_{0}dx)$
is uniquely determined and satisfies an explicit non-linear evolution
equation on $\mathcal{P}(X)$ of the following form: 
\begin{equation}
\frac{d\rho_{t}}{dt}=\frac{1}{\beta}\Delta\rho_{t}-\nabla\cdot(\rho_{t}b[\rho_{t}])\label{eq:evolution equation with drift intro}
\end{equation}
where we have identified $\mu(=\rho dx)$ with its density $\rho$
and $b[\mu]$ is a function on $\mathcal{P}(X)$ taking valued in
the space of vector fields on $X$ defined as minus the gradient of
the differential $dE_{|\mu}$ of $E$ on $\mathcal{P}(X)$ 
\begin{equation}
b[\mu]=-\nabla(dE_{|\mu})\label{eq:drift in terms of energy intro}
\end{equation}
where the differential $dE_{|\mu}$ at $\mu$ is identified with a
function on $X,$ by standard duality (the alternative suggestive
notation $b[\rho]=-\nabla\frac{\partial E(\rho)}{\partial\rho}$ is
often used in the literature). In the kinetic theory literature drift-diffusion
equations of the form \ref{eq:drift in terms of energy intro} are
often called \emph{McKean-Vlasov equation}s. More generally, the results
referred to above hold in the more general setting where the gradient
vector field $-\frac{\partial}{\partial x_{i}}E^{(N)}(x_{1},x_{2},....,x_{N})$
on $X$ appearing in equation \ref{eq:sde general intro} is replaced
by $b[\delta_{N}]$ for a given vector field valued function $v[\mu]$
on $\mathcal{P}(X),$ satisfying appropriate continuity properties.

One of the main aims of the present work is to introduce a new approach
to the propagating of chaos result \ref{eq:convergence empirical measure intro}
for the stochastic dynamics \ref{eq:sde general intro} which exploits
the gradient structure of the equations in question and which applies
under weaker assumptions than the previous results referred to above.
As indicated above our main motivation for weakening the assumptions
comes from the applications to toric Kähler-Einstein metrics described
below. In that case there is a functional $E(\mu)$ on $\mathcal{P}(\R^{n})$
such that

\begin{equation}
\frac{1}{N}E^{(N)}(x_{1},x_{2},....,x_{N})=E(\delta_{N})+o(1),\label{eq:limit of EN intro}
\end{equation}
for a sequence of functionals $E^{(N)}$ which are uniformly Lipschitz
continuous in each variable separately, i.e. there is a constant $C$
such that 
\begin{equation}
|\nabla_{x_{i}}E^{(N)}|\leq C\label{eq:Lip assumption intro}
\end{equation}
 and the error term $o(1)$ tends to zero, as $N\rightarrow\infty$
(for $x_{i}$ uniformly bounded). Moreover, $E^{(N)}$ is $\lambda-$convex
on $X^{N}$ for some real number $\lambda,$ which, by symmetry, means
that the (distributional) Hessian are uniformly bounded from below
for any fixed index $i:$ 
\begin{equation}
\nabla_{x_{i}}^{2}E^{(N)}\geq\lambda I,\label{eq:lambda konvex intro}
\end{equation}
 where $I$ denotes the identity matrix. This implies that there exists
a unique solution to the evolution equation \ref{eq:evolution equation with drift intro}
in the sense of weak gradient flows on the space $\mathcal{P}_{2}(X)$
of all probability measures with finite second moments equipped with
the Wasserstein $L^{2}-$metric \cite{a-g-s}: 

\[
\frac{d\mu_{t}}{dt}=-\nabla F_{\beta}(\mu_{t})
\]
where $F_{\beta}$ is the free energy type functional corresponding
to the macroscopic energy $E(\mu)$ at inverse temperature $\beta:$
\[
F_{\beta}(\mu)=E(\mu)+\frac{1}{\beta}H(\mu),
\]
and where $H(\mu)$ is the Boltzmann entropy of $\mu$ (see Section
\ref{sub:Notation} for notation). 
\begin{thm}
\label{thm:dynamic intro}Suppose that $E^{(N)}$ is a sequence of
symmetric functions on $(\R^{n})^{N}$ satisfying the Main Assumptions
\ref{eq:limit of EN intro}, \ref{eq:Lip assumption intro} and \ref{eq:lambda konvex intro}.
Then, for any fixed positive time $t,$ the empirical measure $\frac{1}{N}\sum\delta_{x_{i}(t)}$
of the system of SDEs \ref{eq:sde general intro} with independent
initial data distributed according to $\mu_{0}\in\mathcal{P}_{2}(\R^{n}$)
converges in law, as $N\rightarrow\infty,$ to the deterministic measure
$\mu_{t}$ evolving by the gradient flow on the Wasserstein space
of the corresponding free energy functional $F_{\beta}$ emanating
from $\mu_{0}.$ 
\end{thm}
It should be stressed that the key point of our approach is that we
do not need to assume that the drift $v[\mu](x)$ defined by formula
\ref{eq:drift in terms of energy intro} has any continuity properties
with respect to $\mu$ or $x$. Or more precisely, even if the $N-$dependent
drift vector field $v^{(N)}$ may very well be smooth for any fixed
$N,$ we do not assume that it is uniformly bounded in $N.$ This
will be crucial in the applications to toric Kähler-Einstein metrics
below.

We recall that if the drift has suitable continuity assumptions, then
the existence of a solution to the drift-diffusion equation \ref{eq:evolution equation with drift intro}
can be established using fix point type arguments \cite{sn}. However,
in our case we have, in general, to resort to the weak gradient flow
solutions provided by the general theory in \cite{a-g-s}, where the
solution $\rho_{t}$ can be characterized uniquely by a differential
inequality called the \emph{Evolutionary Variational Inequality (EVI).}
As shown in \cite{a-g-s} the corresponding solution $\rho_{t}$ satisfies
the drift-diffusion equation\ref{eq:evolution equation with drift intro}
in a suitable weak sense (as follows formally from the Otto calculus
\cite{ot}).

Our approach is inspired by the approach of Messer-Spohn \cite{m-s}
concerning the\emph{ static} problem for the SDEs \ref{eq:sde general intro},
i.e the study of the Boltzmann-Gibbs measure on $X^{N}$ associated
to $E^{(N)}$ at inverse temperature $\beta:$

\[
\mu_{\beta}^{(N)}=\frac{e^{-\beta E^{(N)}}}{Z_{N}}dx^{\otimes N_{k}}
\]
assuming that the normalization function $Z_{N}$ (the partition function)
is finite: 

\[
Z_{N}:=\int_{X^{N}}e^{-\beta E^{(N)}}dx<\infty
\]
From a statistical mechanical point of view this probability measure
describes the microscopic equilibrium distribution at a fixed inverse
temperature $\beta$ and it appears as the law of the large time limit
(for $N$ fixed) of the empirical measures $\delta_{N}(t).$ We thus
get a uniform approach which applies both to the dynamic and the static
setting and which in the latter case leads to the following generalization
of \cite{m-s}:
\begin{thm}
\label{thm:static intro}Suppose that $E^{(N)}$ is uniformly Lipschitz
continuous and convex and satisfies the following uniform properness
assumption: 
\[
E^{(N)}\geq\frac{1}{C}\sum_{i=1}^{N}|x_{i}|-CN
\]
 for some positive constant $C.$ Then the Boltzmann-Gibbs measures
corresponding to $E_{N}$ are well-defined and 
\begin{equation}
\lim_{N\rightarrow\infty}-\frac{1}{N\beta}\log\int e^{-\beta E^{(N)}}dx^{\otimes N}=\inf_{\mathcal{P}_{2}(\R^{n})}F_{\beta}(>-\infty)\label{eq:conv of free enery intro}
\end{equation}
Moreover, the corresponding free energy functional $F_{\beta}$ on
$\mathcal{P}_{2}(\R^{n})$ admits a minimizer $\mu_{\beta}$ and if
it is uniquely determined , then the corresponding empirical measures
on $\R^{n}$ converge in law as $N\rightarrow\infty$ to the deterministic
measure $\mu_{\beta}$ (which, as a consequence, is log concave). 
\end{thm}
It should be stressed that the properness assumption in the previous
theorem which will appear naturally in the setting of toric Kähler-Einstein
metrics below, corresponds to properness wrt the $L^{1}-$Wasserstein
metric on $\mathcal{P}(\R^{n}),$ which is thus strictly weaker than
demanding properness with respect to the $L^{2}-$Wasserstein metric.
But using the convexity assumption on $E_{N},$ we will bypass this
difficulty using Prekopa's inequality and Borell\textquoteright s
lemma.

We briefly point out that the previous theorem is also related to
previous work on lattice spin models such as the Kac model \cite{h-s},
as well as lattice models for random growth of surfaces \cite{fun},
where the large $N-$limit corresponds to the ``thermodynamic limit'',
where the lattice is approximated by a finite volume lattice.

\subsubsection{Idea of the proof of Theorem \ref{thm:dynamic intro} and comparison
with previous results}

The starting point of the proof is the basic fact that the SDEs \ref{eq:sde general intro}
on $X^{N}$ admit a PDE formulation: they correspond to a linear evolution
$\mu_{N}(t)$ of probability measures (or densities) on $X^{N},$
given by the corresponding forward Kolmogorov equation (also called
the Fokker-Planck equation). Given this fact our proof of Theorem
\ref{thm:dynamic intro} proceeds in a variational manner, building
on \cite{a-g-s}: the rough idea to show that the any weak limit curve
$\Gamma(t)$ of the laws 
\[
\Gamma_{N}(t):=(\delta_{N})_{*}\mu_{N}(t)\in\mathcal{P}_{2}(Y),\,\,\,\,\,\,Y=\mathcal{P}_{2}(\R^{n})
\]
 is of the form $\Gamma(t):=\delta_{\mu_{t}},$ where the curve $\mu_{t}$
in $\mathcal{P}_{2}(\R^{n})$ is uniquely determined by a ``dynamic
minimizing property''. To this end we first discretize time, by fixing
a small time mesh $\tau:=t_{j+1}-t_{j}$ and replace, for any fixed
$N,$ the curve $\Gamma_{N}(t)$ with its discretized version $\Gamma_{N}^{\tau}(t_{j}),$
defined by a variational Euler scheme (a ``minimizing movement''
in De Giorgi's terminology) as in \cite{j-k-o,a-g-s}. We then establish
a discretized version of Theorem \ref{thm:dynamic intro} saying that
if, at a given discrete time $t_{j}$ the following convergence holds
in the $L^{2}-$Wasserstein metric 
\[
\lim_{N\rightarrow\infty}\Gamma_{t_{j}}^{N}=\delta_{\mu_{t_{j}}^{\tau}},
\]
 then the convergence also holds at the next time step $t_{j+1}$
(using a variational argument). In particular, since, by assumption,
the convergence above holds at the initial time $0$ it ``propagates''
by induction to hold at any later discrete time. Finally, we prove
Theorem \ref{thm:dynamic intro} by letting the mesh $\tau$ tend
to zero. This last step uses that the very precise error estimates
established in \cite{a-g-s}, for discretization schemes as above,
only depend on a uniform lower bound $\lambda$ on the convexity of
the interaction energies.

Our proof appears to be to rather different from the probabilistic
approaches in \cite{sn,da-g} (and elsewhere) which are based on a
study of non-linear martingales and the recent PDE approach in \cite{mmw}.
As pointed out above these approaches require a Lipschitz control
on the drift vector field $v^{(N)}$ and hence a two-sided uniform
bound on the Hessian of the interaction energy $E^{(N)},$ while we
only require a uniform lower bound.

It may also be illuminating to think about the convergence of $\Gamma_{N}(t)$
towards $\Gamma(t)$ as a kind of a stability result for the sequence
of weak gradient flows on $\mathcal{P}_{2}(Y),$ associated to the
corresponding mean free energies, viewed as functionals on $\mathcal{P}_{2}(Y).$
This situation is somewhat similar to the stability result for gradient
flows on $\mathcal{P}_{2}(H)$ in \cite{a-g-s,a-g-z}, where $H$
is a Hilbert space, but the main difference here is that the underlying
space $Y$ is not a Hilbert space, as opposed to the setting in \cite{a-g-s,a-g-z},
which prevents one from directly applying the error estimates in \cite{a-g-s}
on the space $\mathcal{P}_{2}(Y)$ itself (this analog will be expanded
on in the companion paper \cite{b-o0}).

\subsubsection{\label{sub:Generalizations}Generalizations}

Before continuing with the applications of Theorem \ref{thm:dynamic intro}
(and its static analog) to permanental point process and toric Kähler-Einstein
metrics we want to stress that the assumptions appearing in Theorem
\ref{thm:dynamic intro} may certainly be weakened: 
\begin{itemize}
\item By rescaling $E^{(N)}$ we may as well allow the ``inverse temperature''
$\beta$ appearing in the SDEs \ref{eq:sde general intro} to depend
on $N$ as long as 
\[
\beta_{N}\rightarrow\beta\in[0,\infty],
\]
 as $N\rightarrow\infty.$ In particular, Theorem \ref{thm:dynamic intro}
also applies to $\beta=\infty$ where the evolution equation \ref{eq:evolution equation with drift intro}
becomes a pure transport equation (i.e. with no diffusion). However,
the precise relation to weak solutions becomes much more subtle and
is closely related to the notions of entropy solutions and viscosity
solutions studied in the PDE-literature \cite{la} (as detailed in
\cite{b-o0}). In fact, one may even allow that $\beta_{N}=\infty,$
where the corresponding convergence results yields a deterministic
mean field particle approximation.
\item The assumption \ref{eq:limit of EN intro} in conjunction with the
Lipschitz assumption \ref{eq:Lip assumption intro} may be replaced
by the assumption that the limit of the\emph{ mean energies} corresponding
to $E^{(N)},$ in the sense of statistical mechanics, exists (i.e.
that Proposition \ref{prop:conv of mean energy} below holds) and
that $E^{(N)}$ has a uniform coercivity property. For example, one
can add to the Lipschitz function $E^{(N)}$ any term of the form
$N\mathcal{V}(\delta_{N}),$ for a given coercive $\lambda-$convex
$\mathcal{V}$ functional on $\mathcal{P}_{2}(\R^{n})$ (one then
replaces $E(\mu)$ with $E(\mu)+\mathcal{V}(\mu)).$ 
\item The convexity assumption on $E^{(N)}$ may be replaced by a generalized
convexity property of the corresponding mean energy functional on
$\mathcal{P}_{2}(\R^{nN})^{S_{N}}.$ 
\end{itemize}
These generalizations will be developed in the companion paper \cite{b-o0}.

\subsection{\label{sub:Applications-to-permantel}Applications to permanental
point processes at negative temperature and toric Kähler-Einstein
metrics }

Let $P$ be a convex body in $\R^{n}$ containing zero in its interior
and denote by $P_{\Z}$ the lattice points in $P,$ i.e. the intersection
of the convex body $P$ with the integer lattice $\Z^{n}.$ We fix
an auxiliary ordering $p_{1},...,p_{N}$ of the $N$ elements of $P_{\Z}.$
Given a configuration $(x_{1},...,x_{N})$ of $N$ points on $X$
we denote by $\mbox{Per}(x_{1},...,x_{N})$ the number defined as
the permanent of the rank $N$ matrix with entries $A_{ij}:=e^{x_{i}\cdot p_{j}}:$
\begin{equation}
\mbox{Per}(x_{1},...,x_{N}):=\mbox{Per \ensuremath{(e^{x_{i}\cdot p_{j}})=\sum_{\sigma\in S_{N}}e^{x_{1}\cdot p_{\sigma(1)}+\cdots+x_{N}\cdot p_{\sigma(N)}},}}\label{eq:def of per intro}
\end{equation}
where $S_{N}$ denotes the symmetric group on $N$ letters. This defines
a symmetric function on $\R^{nN}$ which is canonically attached to
$P$ (i.e. it is independent of the choice of ordering of $P_{\Z}).$\footnote{In many body quantum mechanics $\mbox{Per}(x_{1},...,x_{N})$ appears
as the $N-$particle wave function for a bosonic system of $N$ particles
represented by the $N$ wave functions $e^{x\cdot p_{j}},$ i.e. $N$
planar waves with imaginary momenta proportional to $p_{j}.$ } We will consider the large $N$ limit which appears when $P$ is
replaced by the sequence $kP$ of scaled convex bodies, for any positive
integer $k.$ In particular, $N$ depends on $k$ as 
\[
N_{k}=\frac{k^{n}V(P)}{n!}+o(k^{n}),
\]
where $V(P)$ denotes the Euclidean volume of $P.$ In this setting
the interaction energy is defined by 
\begin{equation}
E^{(N_{k})}(x_{1},...,x_{N_{k}})=\frac{1}{k}\log\mbox{Per}(x_{1},...,x_{N_{k}})\label{eq:energy as permanent intro}
\end{equation}
To simplify the notation we will often drop the explicit dependence
of $N$ on $k.$ 

By the results in \cite{b} the assumptions in Theorem \ref{thm:dynamic intro}
hold with 
\[
E(\mu):=-C(\mu),
\]
 where $C(\mu)$ is the Monge-Kantorovich optimal cost for transporting
$\mu$ to the uniform probability measure $\nu_{P}$ on the convex
body $P,$ with respect to the standard symmetric quadratic cost function
$c(x,p)=-x\cdot p.$ Hence, the corresponding free energy functional
may be written as 
\begin{equation}
F_{\beta}(\mu)=-C(\mu)+\frac{1}{\beta}H(\mu)\label{eq:free energy for minus cost intro}
\end{equation}

\begin{thm}
\label{thm:toric dynamic intro}Assume that $\beta>0.$ Then, for
any fixed positive time $t,$ the empirical measure $\frac{1}{N}\sum\delta_{x_{i}}$
of the stochastic process \ref{eq:sde general intro} driven by \ref{eq:energy as permanent intro}
with initial independent data distributed according to a $\mu_{0}\in\mathcal{P}_{2}(\R^{n}$)
converges in law to the deterministic measure $\mu_{t}=\rho_{t}dx$
evolving by the gradient flow on the Wasserstein space, defined by
the functional $F_{\beta}$ (formula \ref{eq:free energy for minus cost intro})
and satisfying the evolution PDE in the distributional sense: 
\begin{equation}
\frac{\partial\rho_{t}}{\partial t}=\frac{1}{\beta}\Delta\rho_{t}+\nabla\cdot(\rho_{t}\nabla\phi_{t})\label{eq:evolut eq for toric ke}
\end{equation}
where $\phi_{t}(x)$ is the unique convex function on $\R^{n}$ solving
the Monge-Ampère equation 
\begin{equation}
\frac{1}{V(P)}\det(\partial^{2}\phi_{t})=\rho_{t}\label{eq:ma eq in system intro}
\end{equation}
(in the weak sense of Alexandrov) normalized so that $\phi(0)=0$
and satisfying the growth condition $\phi(x)\leq\phi_{P}(x),$ where
$\phi_{P}(x):=\sup_{p\in P}p\cdot x.$ 
\end{thm}
Integrating twice reveals that the stationary  equation corresponding
to the evolution PDE in the previous equation may be written as follows
in terms of the convex ``potential'' $\phi:$ 
\begin{equation}
\det(\partial^{2}\phi)=e^{-\beta\phi}\label{eq:static equation ma intro}
\end{equation}
where $\rho_{t}dx=\rho dx:=MA(\phi).$ As shown in \cite{b-b} (generalizing
the seminal result in \cite{w-z}) there is a solution $\phi:=\phi_{\beta}$
to the previous equation iff the origin is the barycenter $b_{P}$
of $P,$ i.e. iff $b_{P}=0$ and then the solution is smooth (see
also \cite{c-k} for generalizations). Moreover, the additive group
$\R^{n}$ acts faithfully by translations on the solution space. Note
that up to replacing $P$ with $\beta^{-1}P$ we may as well assume
that $\beta=1$ and the corresponding static equation 
\begin{equation}
\det(\partial^{2}\phi)=e^{-\phi}\label{eq:k-e eq intro}
\end{equation}
 is precisely the \emph{Kähler-Einstein equation }for a toric Kähler
potential $\phi$ of a Kähler-Einstein metric with positive Ricci
curvature on the toric variety $X_{P}$ corresponding to $P,$ in
the case when $P$ is a rational polytope. More precisely, $X_{P}$
is a toric log Fano variety and $\phi$ corresponds to the Kähler
potential of a Kähler-Einstein metric with conical singularities along
the divisor $X_{P}-\C^{*n}$ ``at infinity'' - the ordinary smooth
Fano case appears when the polytope $P$ is a reflexive Delzant polytope
\cite{w-z,do,b-b}; for some background see Section \ref{sub:The-toric-setting}.

Given the relation to Kähler-Einstein metrics it is natural to ask
if Theorem \ref{thm:toric dynamic intro} has a static analog (as
in Theorem \ref{thm:static intro})? However, it follows from symmetry
considerations involving the action by translations of the additive
group $\R^{n},$ that the corresponding Boltzmann-Gibbs measure 
\[
\mu^{(N_{k})}:=\frac{1}{Z_{N_{k}}}\left(\mbox{Per}(x_{1},...,x_{N})\right)^{-\beta/k}dx^{\otimes N_{k}},
\]
describing a permanental point processes at negative temperature,
is not even well-defined, i.e. the partition function $Z_{N_{k}}$
diverges! This is a reflection of the fact that the static equation
\ref{eq:static equation ma intro} has a a multitude a solutions (due
to the translation symmetry) which from a statistical mechanical point
of view is a sign of a first order phase transition. However, as we
will show, this issue can be bypassed though a symmetry breaking mechanism
where one introduces a ``background potential'' $V(x)$ with appropriate
growth at infinity, dictated by $P,$ i.e. 
\begin{equation}
|V(x)-\phi_{P}(x)|\leq C,\,\,\,\,\phi_{P}(x):=\sup_{p\in P}p\cdot x\label{eq:growth of pot}
\end{equation}
and replace the interaction energy $E^{(N_{k})}(x_{1},...,x_{N_{k}})$
defined by formula \ref{eq:energy as permanent intro} by the convex
combination 
\begin{equation}
E_{V,\gamma}^{(N_{k})}(x_{1},...,x_{N_{k}}):=\gamma\frac{1}{k}\log\mbox{Per}(x_{1},...,x_{N_{k}})+(1-\gamma)\left(V(x_{1})+\cdots+V(x_{N_{k}})\right)\label{eq:def of weighted perm energy}
\end{equation}
for a given parameter $\gamma\in[0,1]$ (which from the point of view
of permanental point processes plays the role of \emph{minus} the
inverse temperature). Then the convergence in Theorem \ref{thm:toric dynamic intro}
still holds with $F$ replaced by 
\[
F_{V,\gamma}(\mu)=-\gamma C(\mu)+(1-\gamma)\int Vd\mu+H(\mu)
\]
and the corresponding static equation now becomes 
\begin{equation}
\det(\partial^{2}\phi)=e^{-(\gamma\phi+(1-\gamma)V)}dx,\label{eq:monge-ampere with gamma}
\end{equation}
which has at most one solution for any given $\gamma\in[0,1]$ and
convex body $P.$ Moreover, if $P$ satisfies the barycenter condition
$b_{P}=0$ then there exists a solution $\phi_{\gamma}$ to the equation
\ref{eq:monge-ampere with gamma} for any $\gamma\in[0,1]$ and as
$\gamma\rightarrow1$ it follows from the results in \cite{w-z,b-b}
that the solutions $\phi_{\gamma}$ converges to a particular solution
to the Kähler-Einstein equation \ref{eq:k-e eq intro}, singled out
by the potential $V.$ 
\begin{thm}
\label{thm:static toric}Let $P$ be a convex body in $\R^{n}$ containing
$0$ in its interior and denote by $b_{P}$ the barycenter of $P.$
For any potential $V$ in $\R^{n}$ satisfying the growth condition
\ref{eq:growth of pot} we have
\begin{itemize}
\item If $b_{P}=0,$ then, for any $\gamma\in[0,1[$ the Gibbs measure $\mu_{V,\gamma}^{(N)}$
corresponding to the energy function $E_{V,\gamma}^{(N_{k})}(x_{1},...,x_{N_{k}}),$
i.e. 
\[
\mu_{\phi_{0},\gamma}^{(N)}:=\frac{1}{Z_{N_{k},\phi_{0},\gamma}}\left(\mbox{Per}(x_{1},...,x_{N})\right)^{-\gamma/k}(e^{-(1-\gamma)V}dx)^{\otimes N_{k}}
\]
is well-defined and equal to the weak limit, as $t\rightarrow\infty$
of the law of the empirical measures for the corresponding SDEs \ref{eq:sde general intro}.
Moreover, as $N\rightarrow\infty$ the corresponding empirical measures
converge in law to the deterministic measure $\mu_{\beta}$ defined
as $\mu_{\gamma}=MA(\phi_{\gamma})$ for the unique (mod $\R)$ solution
$\phi_{\gamma}$ of the equation \ref{eq:monge-ampere with gamma}.
\item More generally, the Gibbs measure above is well-defined precisely
for $\gamma<R_{P},$ where $R_{P}\in[0,1]$ is the following invariant
of $P:$ 
\begin{equation}
R_{P}:=\frac{\left\Vert q\right\Vert }{\left\Vert q-b_{P}\right\Vert },\label{eq:inv R of conv bod}
\end{equation}
 where $q$ is the point in $\partial P$ where the line segment starting
at $b_{P}$ and passing through $0$ meets $\partial P.$ Moreover,
for any such parameter $\gamma$ the corresponding convergence results
still hold.
\end{itemize}
\end{thm}
As indicated in the introduction this result can be viewed as a probabilistic
analog of the seminal Yau-Tian-Donaldson conjecture saying that a
Fano manifold $X$ admits a Kähler-Einstein metric if and only if
$X$ is K-stable (the conjecture has very recently been settled by
Chen-Donaldson-Sun \cite{c-d-s14,c-d-s3}; see also Tian \cite{ti}).
The latter notion is of algebro-geometric nature, but in the toric
setting it is equivalent to the corresponding polytope $P$ having
zero as its barycenter (see Section \ref{sub:The-complex-geometric}
and Corollary \ref{cor:ineq on toric var} for a comparison with the
complex geometric setting).

From a statistical mechanical point of view the critical value $R_{P}$
appearing above can be seen as a real analog of the well-known critical
value appearing in the study of the Keller-Segel equation as well
as in the study of the 2D log gas in \cite{clmp,k}. This connection
will be further expanded on elsewhere \cite{b-l}, but the main point
is that the invariant $R_{P}$ may also be characterized as the sup
over all $\gamma\in]0,\infty[$ such that the free energy type functional
$F_{\gamma}$ is bounded from below (compare \cite{b-b}). 

As we point out in Sections \ref{sub:The-tropical-limit}, \ref{sub:The-toric-setting}
our results also apply to the \emph{tropical} analog of the permanental
setting above, which can be viewed as the tropicalization of the complex
geometric setting on the corresponding toric variety. In the corresponding
deterministic setting (i.e. $\beta_{N}=\infty)$ the particles then
perform zigzag paths in $\R^{n}$ generalizing the extensively studied
Sticky Particle System on $\R$ \cite{ry-si,b-g,ns}. This is closely
related to the Zeldovich model for the formation of large-scale structures
in cosmology; see \cite{f-m-m-s,Br2,br3} (compare the discussion
in Section \ref{sub:Outlook-on-the}).

\subsection{Acknowledgments}

It is a pleasure to thank Eric Carlen for several stimulating discussions
and whose inspiring lecture in the Kinetic Theory seminar at Chalmers
concerning \cite{bcc} drew our attention to the Otto calculus and
the theory of gradient flows on the Wasserstein space. Thanks also
to Luigi Ambrosio for helpful comments on the paper and to Yann Brenier,
Jose Carrillo, Maxime Hauray and Bernt Wennberg for providing us with
references. This work was supported by grants from the Swedish Research
Council, the Knut and Alice Wallenberg Foundation and the European
Research Council.

\subsection{Organization}

In Section \ref{sec:Setup-and-proof} we start by recalling the general
setup that we will need from probability, the theory of Wasserstein
spaces and weak gradient flows and then turn to the proof of Theorem
\ref{thm:dynamic intro} in Section \ref{sub:Propagation-of-chaos}
(starting with the discretized situation). Then the proof of the corresponding
static result, Theorem \ref{thm:static intro} is given. In Section
\ref{sec:Permanental-processes-and}, we go on to apply the previous
general results to the permanental setting, as introduced in Section
\ref{sec:Permanental-processes-and} and its tropical analog. In the
final section we provide on outlook on some relations to conservation
laws, sticky particle type systems and complex geometry. The appendix
recalls the basics of the formal Otto and is included to serve as
a motivation for the material on Wasserstein gradient flows. The rather
lengthy setup and preparatory material in Section \ref{sec:Setup-and-proof}
is due to our effort to make the paper readable to a rather general
audience.

\section{\label{sec:Setup-and-proof}General setup and proof of Theorem \ref{thm:dynamic intro} }

\subsection{\label{sub:Notation}Notation}

Given a topological (Polish) space $Y$ we will denote the integration
pairing between measures $\mu$ on $Y$ (always assumed to be Borel
measures) and bounded continuous functions $f$ by 

\[
\left\langle f,\mu\right\rangle :=\int f\mu
\]
(we will avoid the use of the symbol $d\mu$ since $d$ will usually
refer to a distance function on $Y).$ In case $Y=\R^{D}$ we will
say that a measure\emph{ $\mu$ has a density,} denoted by $\rho,$
if $\mu$ is absolutely continuous wrt Lebesgue measure $dx$ and
$\mu=\rho dx.$ We will denote by $\mathcal{P}(\R^{D})$ the space
of all probability measures and by $\mathcal{P}_{ac}(\R^{D})$ the
subspace containing those with a density. The \emph{Boltzmann entropy
}$H(\rho)$ and \emph{Fisher information} $I(\rho)$ (taking values
in $]-\infty,\infty]$) are defined by 
\begin{equation}
H(\rho):=\int_{\R^{D}}(\log\rho)\rho dx,\,\,\,\,I(\rho)=\int_{\R^{D}}\frac{|\nabla\rho|^{2}}{\rho}dx\label{eq:def of H and I}
\end{equation}
(assuming that $\nabla\rho\in L^{1}(dx)$ and $\rho^{-1}\nabla\rho\in L^{2}(\rho dx)).$
More generally, given a reference measure $\mu_{0}$ on $Y$ the entropy
of a measure $\mu$ relative to $\mu_{0}$ is defined by 
\begin{equation}
H_{\mu_{0}}(\mu)=\int_{X^{N}}\left(\log\frac{\mu}{\mu_{0}}\right)\mu\label{eq:def of rel entropy notation}
\end{equation}
if the probability measure $\mu$ on $X$ is absolutely continuous
with respect to $\mu$ and otherwise $H(\mu):=\infty.$ The relative
Fisher information is defined similarly. Given a lower semi-continuous
(\emph{lsc}, for short) function $V$ on $Y$ and $\beta\in]0,\infty]$
(the ``inverse temperature) we will denote by $F_{\beta}^{V}$ the
corresponding \emph{(Gibbs) free energy functional with potential
$V:$
\begin{equation}
F_{\beta}^{V}(\mu):=\int_{X}V\mu+\frac{1}{\beta}H_{\mu_{0}}(\mu),\label{eq:def of free energy of v notation}
\end{equation}
}which coincides with $\frac{1}{\beta}$ times the entropy of $\mu$
relative to $e^{-V}\mu_{0}.$ In particular, since $H_{\mu_{0}}(\mu)\geq0,$
when $\mu$ and $\mu_{0}$ are probability measures, with equality
iff $\mu=\mu_{0}$ (by Jensen's inequality) the \emph{Boltzmann-Gibbs
measure} 
\[
\frac{e^{-\beta V}}{Z}\mu_{0},\,\,\,\,\,Z:=\int e^{-V}\mu_{0}
\]
 of $V,$ at inverse temperature $\beta$ is the unique minimizer
of $F_{\beta}^{V}$ on the space $\mathcal{P}(Y)$ of probability
measures, under the integrability assumption that $Z<\infty$ (this
is usually called \emph{Gibbs' variational principle}).

\subsection{\label{sub:Wasserstein-spaces-and}Wasserstein spaces and metrics}

We start with the following very general setup. Let $(X,d)$ be a
given metric space, which is Polish, i.e. separable and complete and
denote by $\mathcal{P}(X)$ the space of all probability measures
on $X$ endowed with the \emph{weak topology,} i.e. $\mu_{j}\rightarrow\mu$
weakly in $\mathcal{P}(X)$ iff $\int_{X}\mu_{j}f\rightarrow\int_{X}\mu f$
for any bounded continuous function $f$ on $X$ (this is also called
the\emph{ narrow topology} in the probability literature). The metric
$d$ on $X$ induces $l^{p}-$type metrics on the $N-$fold product
$X^{N}$ for any given $p\in[1,\infty[:$ 
\[
d_{p}(x_{1},...,x_{N};y_{1},...,y_{N}):=(\sum_{i=1}^{N}d(x_{i},y_{i})^{p})^{1/p}
\]
The permutation group $S_{N}$ on $N-$letters has a standard action
on $X^{N},$ defined by $(\sigma,(x_{1},...,x_{N}))\mapsto(x_{\sigma(1)},...,x_{\sigma(N)})$
and we will denote by $X^{(N)}$ and $\pi$ the corresponding quotient
and quotient projection, respectively:

\begin{equation}
X^{(N)}:=X^{N}/S^{N},\,\,\,\,\pi:\,X^{N}\rightarrow X^{(N)}\label{eq:quotient and proj}
\end{equation}
The quotient $X^{(N)}$ may be naturally identified with the space
of all configurations of $N$ points on $X.$ We will denote by $d_{(p)}$
the induced distance function on $X^{(N)},$ suitably normalized:

\[
d_{X^{(N)},l^{P}}(x_{1},...,x_{N};y_{1},...,y_{N}):=\inf_{\sigma\in S_{N}}(\frac{1}{N}\sum_{i=1}^{N}d(x_{i},y_{\sigma(i)})^{p})^{1/p}
\]
The normalization factor $1/N^{1/p}$ ensures that the standard embedding
of $X^{(N)}$ into the space $\mathcal{P}(X)$ of all probability
measures on $X:$ 
\begin{equation}
X^{(N)}\hookrightarrow\mathcal{P}(X),\,\,\,\,(x_{1},..,x_{N})\mapsto\delta_{N}:=\frac{1}{N}\sum\delta_{x_{i}}\label{eq:def of empricical measure}
\end{equation}
(where we will call $\delta_{N}$ the \emph{empirical measure}) is
an isometry when $\mathcal{P}(X)$ is equipped with the\emph{ $L^{p}-$Wasserstein
metric} $d_{W^{p}}$ induced by $d$ (for simplicity we will also
write $d_{W_{p}}=d_{p}):$ 
\begin{equation}
d_{W_{p}}^{p}(\mu,\nu):=\inf_{\gamma}\int_{X\times X}d(x,y)^{p}\gamma,\label{eq:def of wasser}
\end{equation}
 where $\gamma$ ranges over all couplings between $\mu$ and $\nu,$
i.e. $\gamma$ is a probability measure on $X\times X$ whose first
and second marginals are equal to $\mu$ and $\nu,$ respectively
(see Lemma \ref{lem:isometries} below). We will denote $W^{p}(X,d)$
the corresponding\emph{ $L^{p}-$Wasserstein space}, i.e. the subspace
of $\mathcal{P}(X)$ consisting of all $\mu$ with finite $p$ th
moments: for some (and hence any) $x_{0}\in X$ 
\[
\int_{X}d(x,x_{0})^{p}\mu<\infty
\]
We will also write $W^{p}(X,d)=\mathcal{P}_{p}(X)$ when it is clear
from the context which distance $d$ on $X$ is used. 
\begin{rem}
\label{rem:transport}In the terms of the Monge-Kantorovich theory
of optimal transport \cite{v1} $d_{W_{p}}^{p}(\mu,\nu)$ is the optimal
cost to for transporting $\mu$ to $\nu$ with respect to the cost
functional $c(x,p):=d(x,y)^{p}$. Accordingly a coupling $\gamma$
as above is often called a \emph{transport plan} between $\mu$ and
$\nu$ and it said to be defined be a \emph{transport map} $T$ if
$\gamma=(I\times T)_{*}\mu$ where $T_{*}\mu=\nu.$ In particular,
if $X=\R^{n},$ $p=2$ and $\mu$ and $\nu$ are absolutely continuous
with respect to Lebesgue measure, then, by Brenier's theorem \cite{br},
the optimal transport plan $\gamma$ is always defined by a transport
map $T(:=T_{\mu}^{\nu})$ of the form $T_{\mu}^{\nu}=\nabla\phi,$
where $\phi$ is a convex function on $\R^{n}$ (optimizing the dual
Kantorovich functional).
\end{rem}
We recall the following standard 
\begin{prop}
\label{prop:wasserstein conv}A sequence $\mu_{j}$ converges to $\mu$
in the distance topology in $W^{p}(X,d)$ iff $\mu_{j}$ converges
to $\mu$ in the weakly in $\mathcal{P}(X)$ and the $p$ th moments
converge (the latter assumption is automatic if $X$ is compact).
In particular, if $\mu_{j}$ converges to $\mu$ weakly in $\mathcal{P}(X)$
and the $p$ th moments are uniformly bounded, then $\mu_{j}$ converges
to $\mu$ in the distance topology in $W^{p'}(X,d)$ for any $p'<p.$\end{prop}
\begin{proof}
For the first statement see for example \cite[Theorem 7.12]{v1}.
The second statement is certainly also well-known, but for completeness
we include a simple proof. Decompose 
\[
\int_{X}d(x,x_{0})^{p'}\mu_{j}=\int_{\{d(x,x_{0})\leq R\}}d(x,x_{0})^{p'}\mu_{j}+\int_{\{d(x,x_{0})>R\}}d(x,x_{0})^{p'}\mu_{j}
\]
By the assumption and Chebishev's inequality the second terms may
be estimated from above by $C/R^{(p-p')}$ and by the assumption of
weak convergence the first term converges to $\int_{\{d(x,x_{0})\leq R\}}d(x,x_{0})^{p'}\mu$
as $j\rightarrow\infty$ (by taking $f$ to be a suitable regularization
of $1_{\{d(x,x_{0})\leq R\}}d(x,x_{0})^{p'}).$ Finally, letting $R$
tend to infinity concludes the proof.
\end{proof}
Since $Y_{p}:=(W_{p}(X),d_{W_{p}})(:=\mathcal{P}_{p}(X))$ is also
a Polish space we can iterate the previous construction and consider
the Wasserstein space $W_{q}(Y)\subset\mathcal{P}(\mathcal{P}(X))$
that we will write as $W_{q}(\mathcal{P}_{p}(X)),$ which is thus
the space of all probability measures $\Gamma$ on $\mathcal{P}(X)$
such that, for some $\mu_{0}\in W_{p}(X)$ 
\[
\int_{\mathcal{P}(X)}d_{p}(\mu,\mu_{0})^{q}\Gamma<\infty
\]

\begin{lem}
\label{lem:isometries}(Three isometries)
\begin{itemize}
\item The empirical measure $\delta_{N}$ defines an isometric embedding
$(X^{(N)},d_{(p)})\rightarrow\mathcal{P}_{p}(X)$ 
\item The corresponding push-forward map $(\delta_{N})_{*}$ from $\mathcal{P}(X^{(N})$
to $\mathcal{P}(\mathcal{P}(X))$ induces an isometric embedding between
the corresponding Wasserstein spaces $W_{q}(X^{(N)},d_{(p)})$ and
$W_{q}(\mathcal{P}_{p}(X)).$ 
\item The push-forward $\pi_{*}$ of the quotient projection $\pi:X^{N}\rightarrow X^{(N)}$
induces an isometry between the subspace of symmetric measures in
$(W_{q}(X^{N},\frac{1}{N^{1/p}}d_{p})$ and the space $(W_{q}(X^{(N)},d_{(p)})$ 
\end{itemize}
\end{lem}
\begin{proof}
The first statement is a well-known consequence of the Birkhoff-Von
Neumann theorem which gives that for any symmetric function $c(x,y)$
on $X\times X$ we have that if $\mu=\frac{1}{N}\sum_{i=1}^{N}\delta_{x_{i}}$
and $\nu=\frac{1}{N}\sum_{i=1}^{N}\delta_{y_{i}}$ for given $(x_{1},...,x_{N}),(y_{1},...,y_{N})\in X^{N},$
then 
\[
\inf_{\Gamma(\mu,\nu)}\int c(x,y)d\Gamma=\inf_{\Gamma_{N}(\mu,\nu)}\int c(x,y)d\Gamma
\]
where $\Gamma_{N}(\mu,\nu)\subset\Gamma(\mu,\nu)$ consists of couplings
of the form $\Gamma_{\sigma}:=\frac{1}{N}\sum\delta_{x_{i}}\otimes\delta_{y_{\sigma(i)}},$
for $\sigma\in S_{N},$ where $S_{N}$ is the symmetric group on $N$
letters. The second statement then follows from the following general
fact: if $f:(Y_{1},d_{1})\rightarrow(Y_{2},d_{2})$ is an isometry
between two metric spaces, then $f_{*}$ gives an isometry between
$W_{q}(Y_{1},d_{1})$ and $W_{q}(Y_{2},d_{2}).$ This follows immediately
from the definitions once one observes that one may assume that the
coupling $\gamma_{2}$ between $f_{*}\mu$ and $f_{*}\nu$ is of the
form $f_{*}\gamma_{1}$ for some coupling $\gamma_{1}$ between $\mu$
and $\nu.$ The point is that $\gamma$ can be taken to be concentrated
on $f(Y_{1})\times f(Y_{2})$ (since this set contains the product
of the supports of $\mu$ and $\nu)$ and hence one can take $\gamma_{1}:=(f^{-1}\otimes f^{-1})_{*}\gamma_{2}$
where $(f^{-1}\otimes f^{-1})(f(y),f(y')):=(y,y')$ is well-defined,
since $f$ induces a bijection between $Y_{1}$ and $f(Y_{1}).$ Finally,
the last statement follows immediately from the following general
claim applied to $Y=X^{N}$ with $d=\frac{1}{N^{1/p}}d_{X^{N},l^{p}}$
and $G=S_{N}.$ Let $G$ be a compact group acting by isometries on
a metric space $(Y,d)$ and consider the natural projection $\pi:Y\rightarrow Y/G.$
We denote by $d_{G}$ the induced quotient metric on $Y/G.$ The push-forward
$\pi_{*}$ gives a bijection between the space $\mathcal{P}(X)^{G}$
or all $G-$invariant probability measures on $X$ and $\mathcal{P}(X/G).$
The claim is that $\pi_{*}$ induces an isometry between the corresponding
Wasserstein spaces $\mathcal{P}_{q}(X)^{G}$ and $\mathcal{P}_{q}(X/G)$
i.e. $d_{W_{q}}(\mu,\nu)=d_{W_{q}}(\pi_{*}\mu,\pi_{*}\nu)$ if $\mu$
and $\nu$ are $G-$invariant (see \cite[Lemma 5.36]{l-v} Lemma 5.36). 
\end{proof}
Let us also recall the following classical result
\begin{lem}
\label{lem:sanov type}Let $\mu_{0}$ be a probability measure on
$X.$ Then $(\delta_{N})_{*}\mu_{0}^{\otimes N}\rightarrow\delta_{\mu_{0}}$
in $\mathcal{P}(\mathcal{P}(X))$ weakly as $N\rightarrow\infty$ 
\end{lem}
In fact, according to Sanov's classical theorem the previous convergence
results even holds in the sense of large deviations at speed $N$
with rate functional given by the relative entropy functional $H_{\mu_{0}}(\cdot)$
\cite[Theorem 6.2.10]{d-z}. We note that that a (somewhat non-standard)
proof of this classical result can be obtained using the argument
in the proof of Theorem \ref{thm:static intro} applied to $E=0.$

\subsubsection{\label{sub:The-present-setting}The present setting }

We will apply the previous setup to $X=\R^{n}$ endowed with the Euclidean
metric $d.$ Moreover, we will mainly use the case $p=2.$ Then the
corresponding metric $d_{2}$ on $X^{N}$ is the Euclidean metric
on $X^{N}=\R^{nN}.$ Identifying a symmetric (i.e. $S_{N}-$invariant)
probability measure $\mu_{N}$ on $X^{N}$ with a probability measures
on the quotient $X^{(N)}$ (as in Lemma \ref{lem:isometries}) the
second and third point in Lemma \ref{lem:isometries} may (with $q=2)$
be summarized by the following chain of equalities that will be used
repeatedly below:

\begin{equation}
\frac{1}{N}d_{2}(\mu_{N},\mu'_{N})^{2}=d_{(2)}(\mu_{N},\mu'_{N})^{2}=d_{2}(\Gamma_{N},\Gamma'_{N})^{2},\label{eq:distance between symmetric prob measures}
\end{equation}
where $\Gamma_{N}$ and $\Gamma'_{N}$ denote the push-forwards under
$\delta_{N}$ of $\mu_{N}$ and $\mu'_{N},$ respectively.

\subsection{\label{sub:The-Main-Assumptions}The Main Assumptions on the interaction
energy $E^{(N)}$ }

Set $X=\R^{n}$ and denote by $d$ the Euclidean distance function
on $X.$ Throughout the paper $E^{(N)}$ will denote a symmetric,
i.e. $S_{N}-$invariant, sequence of functions on $X^{N}$ and we
will make the following Main Assumptions: 
\begin{enumerate}
\item The functional $E^{(N)}$ is uniformly Lipschitz continuous $(X,d)$
in each variable, or equivalently, under the isometric embedding of
$X^{(N)}$ in $\mathcal{P}(X)$ (by the empirical measure $\delta_{N})$
the sequence $E^{(N)}/N$ extends to define a sequence of functionals
on $\mathcal{P}_{2}(X)$ which are uniformly Lipschitz continuous.
\item The sequence $E^{(N)}/N$ of functions on $\mathcal{P}_{2}(X)$ has
a unique point-wise limit $E(\mu).$
\item The sequence $E^{(N)}$ is $\lambda-$convex on $(X,d)$, uniformly
in $N$.
\end{enumerate}
Note that since $E(\mu)$ above is assumed Lipschitz continuous on
$\mathcal{P}_{2}(X)$ it is uniquely determined by its restriction
to the subspace $\mathcal{P}(X)_{c}$ consisting of all $\mu$ with
compact support.
\begin{lem}
\label{lem:equiv energy as}Assume given a sequence $E^{(N)}$ satisfying
the uniform Lipschitz assumption (1). 
\begin{itemize}
\item Then the second point is (2) is equivalent to point-wise convergence
of $E^{(N)}/N$ towards $E(\mu)$ for any $\mu$ in $\mathcal{P}(X)$
with compact support 
\item The second point (2) implies that\textup{ 
\begin{equation}
\lim_{N\rightarrow\infty}\frac{1}{N}\int E^{(N)}\mu^{\otimes N}=E(\mu)\label{eq:Energy of mu as limit of mean energy}
\end{equation}
for any $\mu\in\mathcal{P}_{2}(X).$}
\item Conversely, if \ref{eq:Energy of mu as limit of mean energy} holds
for any $\mu\in\mathcal{P}(X)$ with compact support then $E^{(N)}/N$
converges towards $E(\mu)$ for any $\mu$ in $\mathcal{P}_{2}(X).$
\end{itemize}
\end{lem}
\begin{proof}
Given $\mu\in\mathcal{P}_{2}(X)$ we define the truncation $\mu_{R}:=\frac{1_{\{d(x,x_{0})\leq R}\mu}{\int1_{\{d(x,x_{0})\leq R}\mu}.$
By the Lip-assumption $\left|(E^{(N)}/N)(\mu)-(E^{(N)}/N)(\mu_{R})\right|\leq Cd_{2}(\mu,\mu_{R}).$
In particular, $a_{N}:=(E^{(N)}/N)(\mu)$ is a uniformly bounded sequence
in $\R.$ Next, letting $N\rightarrow\infty$ gives $\left|a-E(\mu_{R})\right|\leq d_{2}(\mu,\mu_{R})$
for any limit point $a\in\R$ of the sequence $a_{N}.$ Finally letting
$R\rightarrow\infty$ and using the Lip continuity forces $a=E(\mu),$
as desired and hence $(E^{(N)}/N)(\mu)$ converges towards $E(\mu),$
as desired. To prove formula \ref{eq:Energy of mu as limit of mean energy}
we first remark that it follows from the general convergence in Proposition
\ref{prop:conv of mean energy} below that 
\[
\lim_{N\rightarrow\infty}\frac{1}{N}\int E^{(N)}\mu^{\otimes N}=\int_{\mathcal{P}(X)}E(\nu)\Gamma(\nu),
\]
 where $\Gamma$ is a weak limit point of $(\delta_{N})_{*}\mu^{\otimes N}.$
But by Lemma \ref{lem:sanov type}) the limit point is unique and
given by $\Gamma=\delta_{\mu}.$ Hence, the rhs above is equal to
$E(\mu),$ as desired.
\end{proof}

\subsection{\label{sub:The-forward-Kolmogorov}The forward Kolmogorov equation
for the SDEs and the mean free energy $F_{N}$}

Fix a positive integer $N$ and $\beta>0$ (which may depend on $N$
when we will later on let $N\rightarrow\infty).$ Let $(X,g)$ be
a Riemannian manifold and denote by $dV$ the volume form defined
by $g.$ In our case $(X,g)$ will be the Euclidean space $\R^{n}.$
As is well-known, under suitable regularity assumptions, the SDEs
\ref{eq:sde general intro} on $X^{N}$ define, for any fixed $T,$
a probability measure $\eta_{T}$ on the space of all continuous curves
(``sample paths'') in $X^{N},$ i.e. the space of continuous maps
$[0,T]\rightarrow X^{N}$ (see for example \cite{sn} and reference
therein). For $t$ fixed we can thus view $x^{(N)}(t)$ as a $X^{N}-$valued
random variable on the latter probability space. Then its law 
\[
\mu_{t}^{(N)}:=(x^{(N)}(t))_{*}\eta_{t}
\]
 gives a curve of probability measures on $X^{(N)}$ of the form $\mu_{t}^{(N)}=\rho_{t}^{(N)}dV,$
where the density $\rho_{t}^{(N)}$ satisfies the corresponding forward
Kolmogorov equation: 
\begin{equation}
\frac{\partial\rho_{t}^{(N)}}{\partial t}=\frac{1}{\beta}\Delta\rho_{t}^{(N)}+\nabla\cdot(\rho_{t}^{(N)}\nabla E^{(N)}),\label{eq:forward kolm eq}
\end{equation}
which thus coincides with the linear Fokker-Planck equation \ref{eq:linear drift diff}
on $X^{N}$ with potential $V:=E^{(N)}.$ In particular, the law of
the empirical measures $\delta_{N}(t)$ for the SDEs \ref{eq:sde general intro}
can be written as the following probability measure on $\mathcal{P}(X^{N}):$
\[
\Gamma_{N}(t):=(\delta_{N})_{*}\mu_{t}^{(N)},
\]
 where $\delta_{N}$ is the empirical measure defined by formula \ref{eq:def of empricical measure}.

Anyway, for our purposes we may as well forget about the SDEs \ref{eq:sde general intro}
and take the forward Kolmogorov equation \ref{eq:forward kolm eq}
on $X^{N}$ as our the starting point. We will exploit the well-known
fact, going back to \cite{j-k-o} (see Prop \ref{Theorem:f-p as gradient flow}
below) that the latter evolution equation can be interpreted as the
gradient-flow on the Wasserstein space $W_{2}(Y),$ of the functional
\[
F_{\beta}^{(N)}(\mu_{N})=\int_{X^{N}}E^{(N)}\mu_{N}+\frac{1}{\beta}H(\mu_{N}),
\]
 where $H(\cdot)$ is the entropy relative to $\mu_{0}:=dV^{\otimes N}$
(formula \ref{eq:def of rel entropy notation}); occasionally we will
omit the subscript $\beta$ in the notation $F_{\beta}^{(N)}.$

Following standard terminology in statistical mechanics we will call
the scaled functional $F_{N}:=F^{(N)}/N$ the \emph{mean free energy},
which is thus a sum of the \emph{mean energy} $E_{N}(\mu_{N})$ and
the \emph{mean entropy $H_{N}(\mu_{N}):$ 
\[
F_{N}=E_{N}+\frac{1}{\beta}H_{N},
\]
 i.e.} 
\begin{equation}
F_{N}(\mu_{N}):=\frac{1}{N}F^{(N)}(\mu_{N})=\frac{1}{N}\int_{X^{N}}E^{(N)}\mu_{N}+\frac{1}{\beta N}H(\mu_{N}),\label{eq:def of mean free energy}
\end{equation}
Note that it follows immediately from the definition that the mean
entropy is additive: for any $\mu\in\mathcal{P}(X)$

\[
H_{N}(\mu^{\otimes N})=H(\mu)
\]
In case $dV$ is a probability measure it follows immediately from
Jensen's inequality that $H(\mu)\geq0.$ In our Euclidean setting
this is not the case but using that $\int e^{-\epsilon|x|^{2}}dx<\infty$
for any given $\epsilon>0$ one then gets 
\begin{equation}
H(\mu)\geq-\epsilon\int|x|^{2}\mu-C_{\epsilon}\label{eq:lower bound on entropy in terms of first moment}
\end{equation}
As a consequence we have the following 
\begin{lem}
\label{lem:coerc of E gives coerc of F}If the mean energy satisfies
the uniform coercivity property 
\begin{equation}
\frac{1}{N}\int_{X^{N}}E^{(N)}(\mu_{N})\geq-\frac{1}{2\tau_{*}}d_{2}(\mu_{N},\Gamma_{*})^{2}-C\label{eq:coercivity estimate for EN}
\end{equation}
 for some fixed $\tau_{*}>0$ and $\Gamma_{*}\in W_{2}(\mathcal{P}(X))$
and positive constant $C,$ then so does $F^{(N)}/N.$ 
\end{lem}
For example, this is trivially the case under a uniform Lipschitz
assumption on $E^{(N)}$ (as in the Main Assumptions). 
\begin{rem}
The linear forward Kolmogorov equation \ref{eq:forward kolm eq} can
also be viewed as the gradient flow of the\emph{ mean} free energy
$\frac{1}{N}F^{(N)}$ if one instead uses the scaled metric $g_{N}:=\frac{1}{N}g^{\otimes N}$
on $X^{N}.$ Moreover, in our case $E^{(N)}$ will be symmetric, i.e.
$S_{N}-$invariant and hence the flow defined wrt $(X^{N},g_{N})$
descends to the flow defined with respect to $X^{(N)}:=X^{N}/S_{N}$
equipped with the distance function $d_{X^{(N)}}$ defined in Section
\ref{sub:Wasserstein-spaces-and}. Using the isometric embedding defined
by the empirical measure (Lemma \ref{lem:isometries})  we can thus
view the sequence of flows on the sequence of spaces $\mathcal{P}(X^{N})$
as a sequence of flows on the same (infinite dimensional) space $W_{2}(\mathcal{P}(X))$
and this is the geometric motivation for the proof of Theorem \ref{thm:dynamic intro}. 
\end{rem}

\subsection{\label{sub:Gradient-flows-on}Gradient flows on the $L^{2}-$Wasserstein
space and variational discretizations}

In this section we will recall the fundamental results from\cite{a-g-s}
that we will rely on. Let $G$ be a lower semi-continuous function
on a complete metric space $(M,d).$ In this generality there are,
as explained in \cite{a-g-s}, various notions of weak gradient flows
$u_{t}$ for $G$ (or ``steepest descents'') emanating from an initial
point $u_{0}$ in $M,$ symbolically written as 
\begin{equation}
\frac{du_{t}}{dt}=-\nabla G(u_{t}),\,\,\,\,\,\lim_{t\rightarrow0}u(t)=u_{0}\label{eq:gradient flow general}
\end{equation}
The strongest form of weak gradient flows on metric spaces discussed
in \cite{a-g-s} concern $\lambda-$convex functionals $G$ and are
defined by the property that $u_{t}$ satisfies the following \emph{Evolution
Variational Inequalities (EVI) }

\begin{equation}
\frac{1}{2}\frac{d}{dt}d^{2}(u_{t},v)+G(u(t))+\frac{\lambda}{2}d^{2}(\mu_{t},\nu)^{2}\leq G(v)\,\,\,\,\mbox{a.e.\,\,\ensuremath{t>0,\,\,\,\forall v\in M:\,G(v)<\infty}}\label{eq:evi}
\end{equation}
together with the initial condition $\lim_{t\rightarrow0}u(t)=u_{0}$
in $(M,d).$ Then $u_{t}$ is uniquely determined by $u_{0},$ as
shown in \cite[Cor 4.3.3]{a-g-s} and we shall say that $u_{t}$ is
the \emph{EVI-gradient flow} of $G$ emanating from $u_{0}.$ We recall
that $\lambda-$ convexity on a metric space essentially means that
the distributional second derivatives are bounded from below by $\lambda$
along any geodesic segment in $M$ (compare below). When $M$ has
Non-Positive Curvature, NPC (in the sense of Alexandrov) the existence
of a solution $u_{t}$ satisfying the EVI was shown by Mayer \cite{ma}
for any lower-semicontinuos $\lambda-$convex functional, by mimicking
the Crandall-Liggett technique in the Hilbert space setting.

However, in our case $(M,d)$ will be the $L^{2}-$Wasserstein space
$\mathcal{P}_{2}(\R^{d})$ for the space of all probability measures
$\mu$ on $\R^{d}$ which does not have non-positive curvature (when
$d>1).$ Still, as shown in \cite{a-g-s}, the analog of Meyer's result
does hold under the stronger assumption that $G$ be $\lambda-$convex
along any\emph{ generalized} geodesic $\mu_{s}$ in $\mathcal{P}_{2}(\R^{d}).$
For our purposes it will be enough to consider lsc $\lambda-$convex
functionals $G$ with the property that $\mathcal{P}_{2,ac}(\R^{d})$
is weakly dense in $\{G<\infty\}.$ Then the $\lambda-$convexity
of $G$ means (compare \cite[Proposition 9.210]{a-g-s}) that for
any generalized geodesic $\mu_{s}=\rho_{s}dx$ in $\mathcal{P}_{2,ac}(\R^{d})$
the function $G(\rho_{s})$ is continuous on $[0,1]$ and the distributional
second derivatives on $]0,1[$ satisfy 
\[
\frac{d^{2}G(\rho_{s})}{d^{2}s}\geq\lambda,
\]
We recall that a \emph{generalized geodesic} $\mu_{s}$ connecting
$\mu_{0}$ and $\mu_{1}$ in $\mathcal{P}_{2,ac}(\R^{d})$ is determined
by specifying a ``base measure'' $\nu\in\mathcal{P}_{2,ac}(\R^{d})).$
Then $\mu_{s}$ is defined as the following family of push-forwards:
\[
\mu_{s}=\left((1-s)T_{0}+sT_{1}\right)_{*}\nu
\]
 where $T_{i}$ is the optimal transport map (defined with respect
to the cost function $|x-y|^{2}/2)$ pushing forward $\nu$ to $\mu_{i}$
(compare Remark \ref{rem:transport}). 
\begin{rem}
The bona fide Wasserstein geodesics in $\mathcal{P}_{2,ac}(\R^{d})$
are obtained by taking $\nu=\mu_{0}$ (the study of convexity along
such geodesics was introduced by McCann \cite{mcCa}, who called it
displacement convexity). But as shown in \cite{a-g-s} the point of
working with general base measures $\nu$ is that they can be adapted
to the discrete variational scheme for constructing EVI-gradient flows
by taking $\nu=\mu_{t_{j}}$ at the $j$th time step (compare Section
\ref{sub:The-variational-discretization}).
\end{rem}
We will be relying on the following version of Theorem 4.0.4 and Theorem
11.2.1 in \cite{a-g-s}: 
\begin{thm}
\label{thm:existence of evi}Suppose that $G$ is a lsc real-valued
functional on $\mathcal{P}_{2}(\R^{d})$ which is $\lambda-$convex
along generalized geodesics and satisfies the following coercivity
property: there exist constants $\tau_{*},C>0$ and $\mu_{*}\in\mathcal{P}_{2}(\R^{d})$
such that 
\begin{equation}
G(\cdot)\geq-\frac{1}{\tau_{*}}d_{2}(\cdot,\mu_{*})^{2}-C\label{eq:coercivity type condition}
\end{equation}
Then there is a unique solution $\mu_{t}$ to the EVI-gradient flow
of $G,$ emanating from any given $\mu_{0}\in\overline{\left\{ G<\infty\right\} }.$
The flow has the following regularizing effect: $\mu_{t}\in\{\left|\partial G\right|<\infty\}\subset\{G<\infty\}.$
Moreover, $G(\mu_{t})$ and $e^{\lambda t}\left|\partial G\right|^{2}(\mu_{t})$
are decreasing, where $\left|\partial G\right|$ denotes the metric
slope of $G:$ 
\[
\left|\partial G\right|(\mu):=\limsup_{\nu\rightarrow\mu}\frac{(G(\nu)-G(\mu))^{+}}{d(\mu,\nu)}
\]
\end{thm}
\begin{rem}
\label{rem:further prop of evi }Many more properties of the EVI-gradient
flow $\mu_{t}$ are established in \cite{a-g-s}. For example, $\mu_{t}$
defines an absolutely continuous curve $\R\rightarrow\mathcal{P}_{2}(\R^{n})$
(in the sense of metric spaces) which is locally Lipschitz continuous
on $]0,\infty[$ which is a $\lambda-$contracting semi-group. Moreover,
the flows are stable under suitable approximation of the initial data
and the functional $G.$ 
\end{rem}
Under suitably regularity assumptions it shown in \cite{a-g-s} that
the EVI-gradient flow $\mu_{t}=\rho_{t}dx$ furnished by the previous
theorem satisfies Otto's evolution equation (recalled in the appendix)
in the weak sense:
\begin{prop}
\label{prop:evi satisf cont eq}Suppose in addition to the assumptions
in the previous theorem that $\mu_{t}$ has a density $\rho_{t}$
for $t>0.$ Then $\rho_{t}$ satisfies the the continuity equation
\ref{eq:formal gradient flow} in the sense of distributions on $\R^{d}\times\R$
with
\[
v_{t}=-(\partial^{0}G)(\rho_{t}dx),
\]
where $\partial^{0}G$ denotes the minimal subdifferential of $G.$ 
\end{prop}
We recall that under the assumptions in the previous theorem (and
assuming $\{\left|\partial G\right|^{2}<\infty\}\subset\mathcal{P}_{2,ac}(\R^{n}))$
the many-valued \emph{subdifferential} $\partial G$ on the subspace
$\mathcal{P}_{2,ac}(\R^{n})$ is a metric generalization of the (Frechet)
subdifferential Hilbert space theory; by definition, it satisfies
a ``slope inequality along geodesics'': 
\[
(\partial G)(\mu):=\left\{ \xi\in L^{2}(\mu):\,\forall\nu:\,G(\nu)\geq G(\mu)+\left\langle \xi,v\right\rangle _{L^{2}(\mu)}+\frac{\lambda}{2}d_{2}(\nu,\mu)^{2},\,\,\,\,\,v(x):=T_{\mu}^{\nu}(x)-x\right\} 
\]
 where $T_{\mu}^{\nu}$ denotes the optimal transport map between
$\mu$ and $\nu,$ as in Remark \ref{rem:transport} (note that $v$
is the tangent vector field at $0$ of the geodesic $\mu_{s}$ from
$\mu$ to $\nu).$ The \emph{minimal subdifferential} $\partial^{0}G$
on $\mathcal{P}_{2,ac}(\R^{n})$ at $\mu$ is defined as the unique
element in the subdifferential $\partial G$ at $\mu$ minimizing
the $L^{2}-$norm in $L^{2}(\mu);$ in fact, its norm coincides with
the metric slope of $G$ at $\mu$ (in \cite{a-g-s} there is also
a more general notion of extended subdifferential which, however,
will not be needed for our purposes). 
\begin{example}
\label{ex:fisher as norm of subd}In the case when $G=H$ is the Boltzmann
entropy and $\mu$ satisfies $H(\mu)<\infty,$ so that $\mu$ has
a density $\rho,$ we have $(\partial^{0}H)(\mu)=\rho^{-1}\nabla\rho\in L^{2}(\mu)$
and hence 
\[
|\partial H|^{2}(\mu)=I(\rho)
\]
is the Fisher information of $\rho$ (formula \ref{eq:def of H and I});
see \cite[Theorem 10.4.17]{a-g-s}
\end{example}
The following result goes back to McCann \cite{mcCa} (see also \cite{a-g-s}
for various elaborations):
\begin{lem}
\label{lem:generalized conv} The following functionals are lsc and
$\lambda-$convex along any generalized geodesics in $\mathcal{P}_{2}(\R^{d})$:.
\begin{itemize}
\item The ``potential energy'' functional $\mathcal{V}(\mu):=\int V\mu,$
defined by a given lsc $\lambda-$convex and lsc function $V$ on
$\R^{d}$ (and the converse also holds)
\item The functional $\mu\mapsto\int V_{N}\mu^{\otimes N}$ defined by a
given $\lambda-$convex function $V_{N}$ on $\R^{dN}$ and in particular
the ``interaction energy'' functional 
\[
\mathcal{W}(\mu):=\int W(x-y)\mu(x)\otimes\mu(x)
\]
 defined by a given lsc $\lambda-$convex function $W$ on $\R^{d}.$
\item The Boltzmann entropy $H(\mu)$ (relative to $dx$) is lsc and convex
along any generalized geodesics.
\end{itemize}
In particular, for any $\lambda-$convex function $V$ on $\R^{d}$
the corresponding free energy functional $F_{\beta}^{V}$ (formula
\ref{eq:def of free energy of v notation}) is $\lambda-$convex along
generalized geodesics, if $\beta\in]0,\infty].$ 
\end{lem}
Combining the results above we arrive at the following
\begin{thm}
\label{Theorem:f-p as gradient flow} Assume given $\beta\in]0,\infty].$
Let $E(\mu)$ be a lsc functional on $\mathcal{P}_{2}(\R^{d})$ which
is $\lambda-$convex along generalized geodesics and satisfies the
coercivity condition \ref{eq:coercivity type condition}. Then the
EVI-gradient flow $\mu_{t}$ of the corresponding free energy functional
$F_{\beta}:=E+H/\beta$ exists. Moreover, if $\beta<\infty,$ then
$\mu_{t}=\rho_{t}dx,$ where $\rho_{t}$ has finite Boltzmann entropy.
In particular, 
\begin{itemize}
\item If $V$ is a lsc finite $\lambda-$convex function on $\R^{d},$ then
the gradient flow of $F_{\beta}^{V}$ exists (defining a weak solution
of the corresponding forward Kolmogorov equation/Fokker-Planck equation) 
\item If $E(\mu)$ is a Lipschitz continuous functional on $\mathcal{P}_{2}(\R^{d})$
which is $\lambda-$convex along generalized geodesics, then the gradient
flow exists for any initial data $\mu_{0}\in\mathcal{P}_{2}(\R^{n})$
and if $\beta<\infty,$ then $\mu_{t}=\rho_{t}dx,$ where $\rho_{t}$
has finite Boltzmann entropy and Fisher information and the following
continuity equation holds in the distributional sense on $\R^{n}\times\R$
\begin{equation}
\frac{\partial\rho_{t}}{\partial t}=\frac{1}{\beta}\Delta\rho_{t}+\nabla(\rho_{t}v_{t}),\label{eq:cont equa in theorem}
\end{equation}
 where $v_{t}=\partial^{0}E$ is the minimal subdifferential of $E$
at $\mu_{t}=\rho_{t}dx.$ 
\end{itemize}
\end{thm}
\begin{proof}
By the previous Lemma $F_{\beta}$ is also lsc and $\lambda-$convex
and by Lemma \ref{lem:coerc of E gives coerc of F} it also satisfies
the coercivity condition. Hence, the EVI-gradient flow exists according
to Theorem \ref{thm:existence of evi}. Moreover, by the general results
in \cite{a-g-s} $F_{\beta}$ is decreasing along the flow and in
particular locally uniformly bounded from above on $]0,\infty[.$
But, by the coercivity assumption $E>-\infty$ on $\mathcal{P}_{2}(\R^{d})$
and hence it follows that $H(\mu_{t})<\infty.$ The second statement
then follows by the previous lemma and the fact that the coercivity
condition holds: by $\lambda-$convexity $f(x):=v(x)+\lambda|x|^{2}$
is convex and hence $f(x)\geq-C|x|$ for some constant $C,$ proving
coercivity of $v$. To prove the last point first observe that $E(\mu)\geq-A-Bd(\mu,\mu_{0})^{2}<\infty$
on $\mathcal{P}_{2}(\R^{n})$ by the Lip assumption. Since $F_{\beta}(\mu_{t})\leq C$
it follows that $H(\mu_{t})<\infty,$ which in particular implies
that $\mu_{t}$ has a density $\rho_{t}.$ Moreover, by Theorem \ref{thm:existence of evi}
$|\partial F_{\beta}(\mu_{t})|<\infty$ for $t>0.$ But since $E$
is assumed Lip continuous we have $|\partial F_{\beta}(\mu_{t})|<\infty$
iff $|\partial H(\mu_{t})|<\infty,$ which means that $I(\mu_{t})$
has finite Fisher information (see Example r\ref{ex:fisher as norm of subd}).
Finally, the distributional equation follows from Proposition \ref{prop:evi satisf cont eq}.
\end{proof}

\subsubsection{\label{sub:The-variational-discretization}The variational discretization
scheme (``minimizing movements'')}

We recall that the proof of Theorem \ref{thm:existence of evi} in
\cite{a-g-s} uses a discrete approximation scheme introduced by De
Giorgi, called the \emph{minimizing movement} scheme. It can be seen
as a variational formulation of the (back-ward) Euler scheme. Consider
the fixed time interval $[0,T]$ and fix a (small) positive number
$\tau$ (the ``time step''). In order to define the ``discrete
flow'' $u_{j}^{\tau}$ corresponding to the sequence of discrete
times $t_{j}:=j\tau,$ where $t_{j}\leq T$ with initial data $u_{0}$
one proceeds by iteration: given $u_{j}\in M$ the next step $u_{j+1}$
is obtained by minimizing the following functional on $(M,d):=W_{2}(\R^{d}):$
\[
u\mapsto\frac{1}{2\tau}d(u,u_{j})^{2}+G(u)
\]
Next, one defines $u^{\tau}(t)$ for any $t\in[0,T]$ by setting $u^{\tau}(t_{j})=u_{j}^{\tau}$
and demanding that $u^{\tau}(t)$ be constant on $]t_{j},t_{j+1}[$
and right continuous (we are using a slightly different notation than
the one in \cite[Chapter 2]{a-g-s}).

The curve $u_{t}$ is then defined as the large $m$ limit of $u_{t}^{(m)}$
in $(M,d);$ as shown in \cite{a-g-s} the limit indeed exists and
satisfies the EVI \ref{eq:evi} and is thus uniquely determined. More
precisely, the following quantitative convergence result holds (see
\cite[Theorem 4.07, formula 4.024]{a-g-s} and \cite[Theorem 4.09]{a-g-s}):
\begin{thm}
\label{thm:optimal error estimates}Let $G$ be a functional on $\mathcal{P}_{2}(\R^{n})$
satisfying the assumptions in Theorem \ref{thm:existence of evi}
with $\lambda\geq0$. Then 
\[
d^{2}(u^{\tau}(t),u(t))\leq\frac{1}{2}|\tau|^{2}|\partial G|^{2}(u_{0}),
\]
 where $|\partial G|(u_{0})$ denotes the metric slope of $G$ at
$u_{0}.$ If $G$ is only assumed to be $\lambda-$convex for some,
possibly negative, $\lambda$ then 
\[
d(u^{\tau}(t),u(t))\leq C|\tau|(G(u_{0})-\inf G),
\]
for some constant $C$ only depending on $\lambda$ and $T.$ \end{thm}
\begin{rem}
\label{rem:error est}By the last paragraph on page 79 in \cite{a-g-s}
even if $\lambda<0$ one does not need a lower bound on $\inf G$
if one replaces $|\tau|$ with $|\tau|^{1/2},$ as long as $u_{0}$
is assumed to satisfy $G(u_{0})<\infty.$
\end{rem}

\subsection{\label{sub:Propagation-of-chaos}Proof of propagation of chaos in
the discretized setting}

In this section we fix once and for all the time interval $[0,T]$
and the time step $\tau>0.$ We denote by $\mu_{t_{j}}^{(N)}$ the
corresponding discretized minimizing movement of the free energy functional
$F^{(N)}$ on $\mathcal{P}_{2}(X^{N},d_{2})$ with given initial data
$\mu_{t_{j}}^{(N)}.$ The sequence $\mu_{t_{j}}^{(N)}$ is well-defined
according to Theorem \ref{thm:existence of evi} and the Main Assumptions.
Moreover, by the third isometry property in Lemma \ref{lem:isometries}
$\mu_{t_{j}}^{(N)}$ may be identified with the minimizing movement
of the mean free energy functional $F^{(N)}/N$ on $\mathcal{P}_{2}(X^{N},d_{(2)}),$
which in turn embeds isometrically to give a discrete flow $\Gamma_{t_{j}}^{(N)}$
in $W_{2}(\mathcal{P}_{2}(X),d_{2}).$ Similarly, we denote by $\mu_{t_{j}}$
the discretized minimizing movement of the functional $F$ on $\mathcal{P}_{2}(X,)$
with given initial data $\mu_{0}.$
\begin{thm}
\label{thm:induction step}Assume that at time $t_{j}$ 
\[
\lim_{N\rightarrow\infty}(\delta_{N})_{*}\mu_{t_{j}}^{(N)}=\delta_{\mu_{t_{j}}}
\]
in the $L^{2}-$Wasserstein metric. Then, at the next time step $t_{j+1}$
\[
\lim_{N\rightarrow\infty}(\delta_{N})_{*}\mu_{t_{j+1}}^{(N)}=\delta_{\mu_{t_{j+1}}}
\]
in the $L^{2}-$Wasserstein metric.
\end{thm}
We recall that given $\mu_{t_{j}}^{(N)}$ the next measure $\mu_{t_{j+1}}^{(N)}$
is defined as the minimizer of the following functional on $\mathcal{P}(X^{N}):$
\begin{equation}
\frac{1}{N}J_{j+1}^{(N)}(\cdot):=\frac{1}{2\tau}\frac{1}{N}d(\cdot,\mu_{t_{j+1}}^{(N)})^{2}+\frac{1}{N}F^{(N)}(\cdot)\label{eq:def of J functional}
\end{equation}

\subsubsection{Proof of Theorem \ref{thm:induction step}}

We start with the following direct consequence of Proposition \ref{prop:wasserstein conv}
 combined with Lemma \ref{lem:isometries}:
\begin{lem}
\label{lem:dhs}Let $\mu_{N}$ be a sequence of symmetric probability
measures on $X^{N}$ and denote by $\Gamma_{N}:=(\delta_{N})_{*}\mu_{N}$
the corresponding probability measures on $\mathcal{P}(X).$ Assume
that the $d_{q}-$distance of $\Gamma_{N}$ to a fixed element in
the Wasserstein space $W_{q}(\mathcal{P}_{2}(X))$ is uniformly bounded
from above, for some fixed $q\in[1,\infty[.$ Then, after perhaps
passing to a subsequence, there is a probability measure $\Gamma$
in $W_{q}(\mathcal{P}_{2}(X))$ such that
\[
\lim_{N\rightarrow\infty}(\delta_{N})_{*}\mu_{N}=\Gamma
\]
weakly in $\mathcal{P}(X)$ or more precisely in $W_{q'}(\mathcal{P}_{2}(X)$
if $1\leq q'<q$ 
\end{lem}
We next recall the following well-known result about the asymptotics
of the mean entropy (proved in \cite{r-r}; see also Theorem 5.5 in
\cite{h-m} for generalizations). The proof is based on the sub-additivity
properties of the entropy.
\begin{prop}
\label{prop:linf for mean entropy}Let $\mu^{(N)}$ be a sequence
of probability measures on $X^{N}$ such that $(\delta_{N})_{*}\mu_{N}$
converges weakly to $\Gamma\in\mathcal{P}(\mathcal{P}(X)).$ Then 
\end{prop}
\[
\liminf_{N\rightarrow\infty}H^{(N)}(\mu^{N})\geq\int_{\mathcal{P}(X)}H(\mu)\Gamma
\]
We will also use the following result, which generalizes a result
in \cite{m-s} concerning the case when $E_{N}$ is quadratic:
\begin{prop}
\label{prop:conv of mean energy}Let $\mu^{(N)}$ be a sequence of
probability measures on $X^{N}$ such that $\Gamma_{N}:=(\delta_{N})_{*}\mu_{N}$
converges to $\Gamma$ in $W_{1}(\mathcal{P}_{2}(X)).$ Then

\[
\lim_{N\rightarrow\infty}\frac{1}{N}\int_{X^{N}}E^{(N)}\mu^{N}=\int_{\mathcal{P}(X)}E(\mu)\Gamma
\]
\end{prop}
\begin{proof}
Recall that the $L^{1}-$Wasserstein distance $d_{1}$ on $\mathcal{P}(Y),$
where $Y=\mathcal{P}_{2}(X),$ admits the following dual representation:
\[
d(\mu,\nu)=\sup_{u\in Lip_{1}}\int u(\mu-\nu)
\]
where $u$ ranges over all Lip-functions on $Y$ with Lip-constant
one. By assumption 
\begin{equation}
d_{1}(\Gamma_{N},\Gamma)\rightarrow0.\label{eq:conv wrt l1 wasser}
\end{equation}
Using the empirical measure $\delta_{N}$ we identify $N^{-1}E^{(N)}$
with a uniformly Lipschitz continuous sequence of functions on $\mathcal{P}(X)$
which by the Main Assumptions point-wise to to $E(\mu).$ First observe
that since $N^{-1}E^{(N)}$ is uniformly Lipschitz continuous we have
that 
\[
\lim_{N\rightarrow0}\int_{\mathcal{P}(X)}N^{-1}E^{(N)}(\Gamma_{N}-\Gamma)=0
\]
using \ref{eq:conv wrt l1 wasser} combined with the dual representation
of the $L^{1}-$Wasserstein distance. Hence, 
\[
\lim_{N\rightarrow\infty}\frac{1}{N}\int_{X^{N}}E^{(N)}\mu^{N}=\lim_{N\rightarrow\infty}\int_{\mathcal{P}(X)}N^{-1}E^{(N)}\Gamma=\int_{\mathcal{P}(X)}E(\mu)\Gamma
\]
as desired (using the dominated convergence theorem in the last step,
which applies thanks to the bound $\left|N^{-1}E^{(N)}\right|\leq A+Bd_{2}$
resulting from the Main Assumptions).
\end{proof}
Next we turn to the asymptotics of the distances, establishing the
following key property:
\begin{prop}
\label{prop:conv of mean dist}Assume that a sequence $\nu_{N}$ of
symmetric probability measures on $X^{N}$ satisfies 
\[
\lim_{N\rightarrow\infty}(\delta_{N})_{*}\nu_{N}=\delta_{\nu}
\]
 in the distance topology in $W_{2}(\mathcal{P}_{2}(X)).$ Then any
sequence $\mu_{N}$ such that $(\delta_{N})_{*}\mu_{N}$ converges
weakly to $\Gamma\in\mathcal{P}(\mathcal{P}(X))$ satisfies 
\[
\liminf_{N\rightarrow\infty}\frac{1}{N}d(\mu_{N},\nu_{N})^{2}\geq\int_{\mathcal{P}(X)}d(\mu,\nu)^{2}\Gamma(\mu)
\]
and equality holds iff $(\delta_{N})_{*}\mu_{N}$ converges to $\Gamma$
in the distance topology in $W_{2}(\mathcal{P}_{2}(X)).$\end{prop}
\begin{proof}
Consider the isometry 
\[
\delta_{N}:\,(X^{(N)},d_{X^{(N)}})\hookrightarrow(\mathcal{P}(X),d_{W})\,\,\,\,(x_{1},..,x_{N})\mapsto\delta_{N}:=\frac{1}{N}\sum\delta_{x_{i}}
\]
defined in terms of the $L^{2}-$distances. We equip the space $\mathcal{P}(\mathcal{P}(X))$
with the $L^{2}-$Wasserstein (pre-)metric $d$ induced from distance
$d_{W}$ on $\mathcal{P}(X),$ i.e. we consider the subspace $W_{2}(\mathcal{P}(X)).$
By Lemma \ref{lem:isometries}
\[
\frac{1}{N}d(\mu_{N},\nu_{N})^{2}=d((\delta_{N})_{*}\mu_{N},(\delta_{N})_{*}\nu_{N})^{2}
\]
We now first assume that $(\delta_{N})_{*}\mu_{N}$ converges to $\Gamma$
in the $d-$distance topology in $W_{2}(\mathcal{P}_{2}(X)).$ Then
the ``triangle inequality'' for $d$ immediately gives 
\[
\lim_{N\rightarrow\infty}d((\delta_{N})_{*}\mu_{N},(\delta_{N})_{*}\nu_{N})^{2}=d(\Gamma,\delta_{\nu})^{2}.
\]
Next we will use the following simple general fact for the Wasserstein
distance on $\mathcal{P}(Y,d):$
\[
d(\mu,\delta_{y_{0}})^{2}=\int d(y,y_{0})^{2}\mu(y)
\]
which follows from the fact that the only coupling between $\mu$
and $\delta_{y_{0}}$ is the product $\mu\otimes\delta_{y_{0}}.$
Applied to $Y=\mathcal{P}(X)$ this gives 
\[
d((\delta_{N})_{*}\mu_{N},\delta_{\nu})^{2}=\int_{\mathcal{P}(X)}d(\mu,\nu)^{2}\Gamma(\mu)
\]
 which concludes the proof using that $d(\delta_{\mu},\delta_{\nu})=d(\mu,\nu)$
by the general fact above. More generally, if $(\delta_{N})_{*}\mu_{N}$
is only assumed to converge to $\Gamma$ weakly in $\mathcal{P}(\mathcal{P}(X)),$
then the lower semi-continuity of the Wasserstein distance function
wrt the weak topology instead gives 
\[
\liminf_{N\rightarrow\infty}\frac{1}{N}d(\mu_{N},\nu_{N})^{2}\geq\int_{\mathcal{P}(X)}d(\mu,\nu)^{2}\Gamma(\mu)
\]
Finally, if equality holds above, then, by the previous arguments,
\[
\lim_{N\rightarrow\infty}\int_{\mu\in\mathcal{P}(X)}d(\mu,\nu)^{2}(\delta_{N})_{*}\mu_{N}=\int_{\mathcal{P}(X)}d(\mu,\nu)^{2}\Gamma(\mu)
\]
(i.e. the ``second moments of $(\delta_{N})_{*}\mu_{N}$ converge
to the second moments of $\Gamma$) and then it follows from Proposition
\ref{prop:wasserstein conv} that $(\delta_{N})_{*}\mu_{N}$ converges
to $\Gamma$ in the distance topology in $W_{2}(\mathcal{P}(X)).$ 
\end{proof}

\subsubsection{\label{sub:Conclusion-of-the induction}Conclusion of the proof of
Theorem \ref{thm:induction step}}

Without loss of generality we may set $\beta=1.$ We start by observing
that for any fixed $\mu$ in $\mathcal{P}(X)$ we have, by the defining
property of $\mu_{t_{j+1}}^{(N)},$ that 
\[
J_{j+1}^{(N)}(\mu_{t_{j+1}}^{(N)})/N\leq J_{j+1}^{(N)}(\mu^{\otimes N})/N
\]
where the rhs converges, by the propositions above, to $J_{j+1}(\mu)$
as $N\rightarrow\infty,$ where $J_{j+1}(\mu)=\frac{1}{2h}d(\mu,\mu_{j})^{2}+F(\mu).$
In particular, taking $\mu=\mu_{j+1}$ gives 
\begin{equation}
\limsup_{N\rightarrow\infty}J_{j+1}^{(N)}(\mu_{t_{j+1}}^{(N)})/N\leq J_{j+1}(\mu_{j+1})\label{eq:upper bound in induction}
\end{equation}
 where $\mu_{j+1}$ is the unique minimizer of $J_{j+1}.$

Next we consider the lower bound. By the minimizing property of $\mu_{t_{j+1}}^{(N)}$
we have a uniform control on the $d_{2}-$distance: 
\begin{equation}
d_{2}((\delta_{N})_{*}\mu_{t_{j+1}}^{(N)},(\delta_{N})_{*}\mu_{t_{j}}^{(N)})^{2}=\frac{1}{N}d_{2}(\mu_{t_{j+1}}^{(N)},\mu_{t_{j}}^{(N)})^{2}\leq C\label{eq:proof of induction step distance control}
\end{equation}
Indeed, the minimizing property together with the previous bound gives
\[
\frac{1}{\tau}d_{2}((\delta_{N})_{*}\mu_{t_{j+1}}^{(N)},(\delta_{N})_{*}\mu_{t_{j}}^{(N)})^{2}\leq C-\frac{1}{N}F^{(N)}(\mu_{t_{j+1}}^{(N)})
\]
Hence, it is enough to verify that the uniform coercivity property
\ref{eq:coercivity estimate for EN} holds. But this follows the uniform
Lipschitz assumption on $E^{(N)}.$ 

Now, it follows from the induction assumption and the triangle inequality
for $d$ that $\mu_{t_{j+1}}^{(N)}$ satisfies the assumptions of
Lemma \ref{prop:wasserstein conv}. Accordingly, we may, after passing
to a subsequence, assume that $\mu_{N}:=\mu_{t_{j+1}}^{(N)}$ converges
as in Lemma \ref{lem:dhs}, or more precisely that 
\[
(\delta_{N})_{*}\mu_{t_{j+1}}^{(N)}\rightarrow\Gamma
\]
 in $W_{1}(\mathcal{P}(X)),$ for some $\Gamma\in W_{2}(\mathcal{P}(X)).$
It then follows from Propositions \ref{prop:linf for mean entropy},
\ref{prop:conv of mean energy} and \ref{prop:conv of mean dist}
that 
\begin{equation}
\liminf_{N\rightarrow\infty}J_{j+1}^{(N)}(\mu_{t_{j+1}}^{(N)})/N\geq\int d\Gamma(\mu)J_{j+1}(\mu)\label{eq:lower bound in induction}
\end{equation}
Combining the previous lower bound with the upper bound \ref{eq:upper bound in induction}
and using that $\mu_{j+1}$ is the unique minimizer of $J_{j+1}$
then forces $\Gamma=\delta_{\mu_{j+1}}$ and 
\begin{equation}
\lim_{N\rightarrow\infty}J_{j+1}^{(N)}(\mu_{t_{j+1}}^{(N)})/N=J_{j+1}(\mu)\label{eq:conv of J}
\end{equation}
But this means that 
\[
\lim_{N\rightarrow\infty}(\delta_{N})_{*}\mu_{t_{j+1}}^{(N)}=\delta_{\mu_{t_{j+1}}}
\]
 weakly in $\mathcal{P}(X)$ and by the equality \ref{eq:conv of J}
that 
\[
\lim_{N\rightarrow\infty}d(\delta_{N})_{*}\mu_{t_{j+1}}^{(N)},\delta_{\mu_{t_{j+1}}})=d(\delta_{\mu_{t_{j+1}}},\delta_{\mu_{j}}).
\]
 But then it follows from Proposition  \ref{prop:conv of mean dist}
(applied to $\nu=\delta_{\mu_{j}})$ that $(\delta_{N})_{*}\mu_{N}$
converges to $\Gamma$ in the distance topology in $W_{2}(\mathcal{P}(X)),$
as desired.

\subsection{Convergence in the non-discrete setting: proof of Theorem \ref{thm:dynamic intro}}

We first assume that $\lambda\geq0.$ By the Main Assumptions the
limiting free energy functional $F(\mu):=E(\mu)+H(\mu)$ is also Lipschitz
continuous and $\lambda-$convex along generalized geodesics in $\mathcal{P}_{2}(X).$
Indeed, by \ref{eq:Energy of mu as limit of mean energy} $E(\mu)$
is the limit of the mean energy functionals $\mu\mapsto\int_{X^{N}}E^{(N)}/N\mu^{\otimes N}$
which are $\lambda-$convex along generalized geodesics, since $E^{(N)}/N$
is assumed $\lambda-$convex (see Lemma \ref{lem:generalized conv}).
In particular, by Theorem \ref{Theorem:f-p as gradient flow} the
gradient flow $\mu_{t}$ of $F$ emanating from a given $\mu_{0}\in\mathcal{P}(X)$
exists and is uniquely determined in the sense of Theorem \ref{thm:existence of evi}.
We let $\Gamma_{t}:=\delta_{\mu_{t}}$ be the corresponding flow on
$\mathcal{P}_{2}(\mathcal{P}_{2}(X)).$

Consider the fixed time interval $[0,T]$ and fix a small time step
$\tau>0.$ Denote by $\mu^{\tau}(t)$ the discretized minimizing movement
of $F(\mu)$ with time step $\tau$ and set $\Gamma_{t}^{\tau}:=\delta_{\mu_{t}^{\tau}}.$
For any fixed $t\in]0,T[$ we then have, by the triangle inequality, 

\[
d(\Gamma_{N}(t),\Gamma(t))\leq d(\Gamma_{N}(t),\Gamma_{N}^{\tau}(t))+d(\Gamma(t),\Gamma^{\tau}(t))+d(\Gamma_{N}^{\tau}(t),\Gamma^{\tau}(t))
\]
By the isometry property in Lemma \ref{lem:isometries} and the assumed
convexity properties we have, by Theorem \ref{thm:optimal error estimates},
that $d(\Gamma_{N}(t),\Gamma_{N}^{\tau}(t))\leq C\tau$ (uniformly
in $N)$ and $d(\Gamma(t),\Gamma^{\tau}(t))\leq\tau C.$ Moreover,
by Theorem \ref{thm:induction step} $\lim_{N\rightarrow\infty}d(\Gamma_{N}^{\tau}(t),\Gamma^{\tau}(t))=0$
for any fixed $\tau.$ Hence, letting first $N\rightarrow\infty$
and then $\tau\rightarrow0$ gives $\lim_{N\rightarrow\infty}d(\Gamma_{N}(t),\Gamma(t))=0,$
which concludes the proof.

In the case when $\lambda\leq0$ the previous argument still applies
(with the error $O(\tau)$ replaced by $\mathcal{O}(\tau^{1/2})$
according to Remark \ref{rem:error est}.

\subsection{Proof of Theorem \ref{thm:static intro} (the static case)}

By Gibbs variational principle (which follows immediately from Jensen's
lemma) $\mu^{(N)}$ is a minimizer of the mean free energy $F^{(N)}.$
In particular, for any fixed $\mu\in\mathcal{P}(\R^{n})$ 
\begin{equation}
\frac{1}{N}F^{(N)}(\mu^{(N)})\leq\frac{1}{N}F^{(N)}(\mu^{\otimes N})=H(\mu)+\int\frac{1}{N}E^{(N)}\mu^{\otimes N}\leq C,\label{eq:upper bound on mean free in proof static}
\end{equation}
where the last inequality is obtained by taking $\mu$ to be any measure
with compact support. Next observe that by the properness assumption
on $F^{(N)}$ this gives 

\[
\int_{(\R^{n})^{N}}\frac{\text{1}}{N}\sum_{i=1}^{N}|x_{i}|\mu^{(N)}\leq A\frac{1}{N}F^{(N)}(\mu^{(N)})-B\leq AC-B
\]
However, on order to apply Proposition \ref{prop:conv of mean energy}
we would rather need a bound on the $p-$moments for some $p>1:$ 

\begin{equation}
\int_{(\R^{n})^{N}}\frac{\text{1}}{N}\sum_{i=1}^{N}|x_{i}|^{p}\mu^{(N)}\leq C_{p}\label{eq:bound on p th moment}
\end{equation}
But using the convexity assumption this follows automatically from
the bound on the first moments using the following well-consequence
of Borell's lemma \cite{bo} (see \cite[Appendix III]{mil-s}), which
gives a ``Reversed Hölder's inequality'':
\begin{lem}
\label{lem:borell}Let $\mu$ be a log concave measure, i.e. $\mu=e^{-\phi}dx$
for some convex functions $\phi.$ Then, for any $q>p:$ 
\[
\left(\int_{\R^{n}}|x|^{q}\mu\right)^{1/q}\leq C_{p,q,n}\left(\int_{\R^{n}}|x|^{p}\mu\right)^{1/p},
\]
 where the constant $C_{p,q,n}$ is independent of $\mu$
\end{lem}
To prove the bound \ref{eq:bound on p th moment} we first observe
that for any $p$

\[
\int_{(\R^{n})^{N}}\frac{\text{1}}{N}\sum_{i=1}^{N}|x_{i}|^{p}\mu^{(N)}=\int_{\R^{n}}|x|^{p}(\mu^{(N)})_{1},
\]
 where $(\mu^{(N)})_{1}$ denotes the ``first marginal'' of $\mu^{(N)},$
i.e. its push-forward under the natural projection $(\R^{n})^{N}\rightarrow\R^{n}$
onto the first factor. But, by assumption $\mu^{(N)}$ is log concave
and hence, by the Prekopa inequality \cite{v1}, so is its first marginal
$(\mu^{(N)})_{1}.$ Applying the previous lemma thus gives 
\[
\int_{\R^{n}}|x|^{p}(\mu^{(N)})_{1}\leq C_{p,n}
\]
 for any $p\geq1$ where $C_{p,n}$ is independent of $N.$ In particular,
by Proposition \ref{prop:wasserstein conv} combined with the isometric
embedding in Lemma \ref{lem:isometries} we get that $(\delta_{N})_{*}\mu^{(N)}:=\Gamma^{(N)}$
converges to some $\Gamma$ in the distance topology in $W_{p}(\mathcal{P}(X))$
for any $p\in[1,\infty[.$ Applying this to $p=1$ and invoking Proposition
\ref{prop:linf for mean entropy} thus gives the convergence of the
mean energy: 
\[
\lim_{N\rightarrow\infty}\frac{1}{N}\int_{X^{N}}E^{(N)}(\mu^{N})=\int_{\mathcal{P}(X)}E(\mu)\Gamma
\]
We then deduce, using the asymptotics of the entropy in Prop \ref{prop:linf for mean entropy},
precisely as in the proof of Theorem \ref{thm:induction step}, that 

\[
\int_{\mathcal{P}(X)}F(\mu)\Gamma(\mu)\leq\liminf_{N\rightarrow\infty}\frac{1}{N}F^{(N)}(\mu^{(N)})\leq\frac{1}{N}F^{(N)}(\mu)
\]
Next we get, using the upper bound \ref{eq:upper bound on mean free in proof static}
and Proposition \ref{prop:linf for mean entropy} applied to $\mu_{N}=\mu^{\otimes N},$
that 
\[
\int_{\mathcal{P}(X)}F(\mu)\Gamma\leq\liminf_{N\rightarrow\infty}\frac{1}{N}F^{(N)}(\mu^{(N)})=F(\mu_{*})=\inf_{W_{2}(X)}F(\mu)
\]
Hence, it must be that $\Gamma=\delta_{\mu_{*}}$ which concludes
the proof of the convergence assuming the existence and uniqueness
of the minimizer $\mu_{*}.$ In fact, the existence of $\mu_{*}$
also follows from the previous argument: indeed, since, by well-known
properties of the entropy \cite{m-s}, we have $H^{(N)}(\mu)/N\geq H(\mu_{1}^{(N)})$
the previous argument gives that any limit point $\mu_{*}$ of the
sequence $\mu_{1}^{(N)}$ in $\mathcal{P}(X)$ (which exists by tightness
and, as explained above, is in $W_{p}(X)$ for any $p)$ satisfies
\[
F(\mu_{*})\leq\liminf_{N\rightarrow\infty}\frac{1}{N}F^{(N)}(\mu^{(N)})=F(\mu_{*})=\inf_{W_{2}(X)}F(\mu)
\]
and hence minimizes $F$ on $W_{2}(X).$ Finally, since $\mu_{1}^{(N)}$
is log concave so is the limit $\mu_{*}.$ As for the convergence
\ref{eq:conv of free enery intro} it follows immediately from the
formula $-\frac{1}{\beta}\log Z_{N}=F_{\beta}^{(N)}(\mu^{(N)})$ and
the proof of the above does not use neither existence nor uniqueness
of a minimizer of $F.$

\section{\label{sec:Permanental-processes-and}Permanental processes and toric
Kähler-Einstein metrics}

\subsection{Permanental processes: setup}

Let $P$ be a convex body in $\R^{n}$ containing zero in its interior
and denote by $\nu_{P}$ the corresponding uniform probability measure
on $P,$ i.e. $P=1_{P}d\lambda/V(P),$ where $d\lambda$ denotes Lebesgue
measure and $V(P)$ is the Euclidean volume of $P.$ Setting $P_{k}:=P\cap(\Z/k)^{n},$
we let $N_{k}$ be the number of points in $P_{k}$ and fix an auxiliary
ordering $p_{1},...,p_{N_{k}}$ of the $N_{k}$ elements of $P_{k}.$
Given a configuration $(x_{1},...,x_{N_{k}})$ of points on $X:=\R^{n}$
we set 
\begin{equation}
E^{(N_{k})}(x_{1},...,x_{N_{k}}):=\frac{1}{k}\log\sum_{\sigma\in S_{N_{k}}}e^{k(x_{1}\cdot p_{\sigma(1)}+\cdots+x_{N}\cdot p_{\sigma(N_{k})})},\label{eq:permanental energy text}
\end{equation}
which, as explained in the introduction of the paper, can be written
as the scaled logarithm of a permanent. To simplify the notation we
will often drop the subscript $k$ and simply write $N_{k}=N,$ since
anyway $N\rightarrow\infty$ iff $k\rightarrow\infty.$ We will denote
by $C(\mu,\nu)$  the Monge-Kantorovich optimal cost for transport
between the probability measures $\mu$ and $\nu,$ with respect to
the standard symmetric quadratic cost function $c(x,p)=-x\cdot p:$
\begin{equation}
C(\mu,\nu):=\inf_{\gamma}-\int_{X\times X}x\cdot p\gamma,\label{eq:cost to convex body}
\end{equation}
 where the $\gamma$ ranges over all couplings (transport plans) between
$\mu$ and $\nu.$ 
\begin{prop}
\label{prop:main assump in perman case}The Main Assumptions for $E^{(N)}$
are satisfied with $\lambda=0$ and $E(\mu)=-C(\mu,\nu_{P}),$ Equivalently,
formulated in terms of the Wasserstein $L^{2}-$distance 
\begin{equation}
E(\mu)=-\frac{1}{2}d_{W_{2}}(\mu,\nu_{P})^{2}+\frac{1}{2}\int x^{2}d\mu+c_{P},\,\,\,\,\,c_{P}:=\frac{1}{2}\int p^{2}\nu_{P}\label{eq:cost as distance}
\end{equation}
In particular, $-C(\cdot,\nu_{P})$ is convex along generalized geodesics.\end{prop}
\begin{proof}
This follows from the results in \cite{b}. In fact, the first and
second point follows immediately from basic fact that if $\phi_{s}$
is a family of smooth convex functions on $\R^{n}$ and $\nu$ a probability
measure on the parameter space, then $\phi:=\log\int e^{\phi_{s}}\nu(s)$
is also convex and $\nabla\phi$ is contained in the convex support
of $\{\nabla\phi_{s}\},$ which in the present case is contained in
$kP.$ Hence, $\nabla_{x_{i}}E^{(N)}\in P$ which is uniformly bounded,
since $P$ is a convex body and in particular bounded. Finally, the
convergence of $E^{(N)}$ was shown in \cite{b} for $\mu$ with compact
support (which is enough by Lemma \ref{lem:equiv energy as}). The
convexity of $-C(\cdot,\nu_{P})$ then follows from Lemma \ref{lem:generalized conv}.
Equivalently, this means that $-\frac{1}{2}d_{W_{2}}(\mu,\nu_{P})^{2}$
is $-1-$convex. In fact, as shown in \cite{a-g-s} using a different
argument $-\frac{1}{2}d_{W_{2}}(\cdot,\nu)^{2}$ is $-1-$convex for
any fixed $\nu\in\mathcal{P}_{2}(\R^{n}).$
\end{proof}
Next, we recall that the \emph{Monge-Ampère measure} $MA(\phi)$ of
a convex function $\phi$ on $\R^{n}$ is defined by the property
that, for a given Borel set $E,$ 
\[
\int_{E}MA(\phi):=\int_{(\partial\phi)(E)}d\lambda,
\]
 where $d\lambda$ denotes Lebesgue measure and $\partial\phi$ denotes
the subgradient of $\phi$ (which defines a multivalued map from $\R^{n}$
to $\R^{n}).$ In particular, if $\phi\in C_{loc}^{2},$ then 
\[
MA(\phi)=\det(\partial^{2}\phi)dx,
\]
 where $\partial^{2}\phi$ denotes the Hessian matrix of $\phi.$
We will denote by $\mathcal{C}_{P}$ the space of all convex functions
$\phi$ on $\R^{n}$ whose subgradient $\partial\phi$ satisfies 
\[
(\partial\phi)(\R^{n})\subset P
\]
and we will say that $\phi$ is \emph{normalized }if $\phi(0)=0.$
By the convexity of $\phi$ the gradient condition above equivalently
means that $\phi$ grows as most as the support function $\phi_{P}$
of $P,$ where $\phi_{P}(x):=\sup_{p\in P}p\cdot x$

By Brenier's theorem \cite{br}, given $\mu=\rho dx$ in $\mathcal{P}_{2}(\R^{n})$
there exists a unique normalized $\phi\in\mathcal{C}_{P}$ such that
\begin{equation}
MA(\phi)=\rho dx,\label{eq:ma eq in text}
\end{equation}
 which equivalently means that the corresponding $L^{\infty}-$map
$\nabla\phi$ from $\R^{n}$ to $P$ satisfies 
\[
(\nabla\phi)_{*}\mu=\nu_{P}
\]
Given the previous proposition we can use the differentiability result
in \cite{a-g-s} for the Wasserstein $L^{2}-$distance to get the
following
\begin{lem}
\label{lem:The-minimal-subdiffer is ma}The minimal subdifferential
of $-C(\cdot,\nu_{P})$ on the subspace $\mathcal{P}_{2,ac}(\R^{n})$
of all probability measures in $\mathcal{P}_{2}(\R^{n})$ which are
absolutely continuous wrt $dx,$ may, at a given point $\rho dx,$
be represented by the $L^{\infty}-$vector field $\nabla\phi,$ where
$\phi$ is the unique normalized solution in $\mathcal{C}_{P}$ to
the equation \ref{eq:ma eq in text}.\end{lem}
\begin{proof}
Given formula \ref{eq:cost as distance} this follows immediately
from Theorem 10.4.12 in \cite{a-g-s} and the fact that if $\mathcal{P}_{2,ac}(\R^{n}),$
then Brenier's theorem  gives that the optimal transport plan (coupling)
from $\R^{n}$ to $P$ realizing the infimum defining $d_{W_{2}}(\mu,\nu_{P})^{2}$
is given by the $L^{\infty}-$map $\nabla\phi,$ where $\phi$ solves
the equation \ref{eq:ma eq in text}. Since the barycentric projection
appearing in Theorem 10.4.12 in \cite{a-g-s} for the transport plan
defined by a transport map gives back the transport map (see \cite[Th, 12.4.4]{a-g-s}
this concludes the proof.  See also \cite{b-b} for a direct variational
proof of Brenier's theorem which can be seen as the real analogue
of the variational approach to\emph{ complex }Monge-Ampère equations
in \cite{bbgz}.
\end{proof}

\subsection{Existence of the gradient flow for $F_{\beta}(\mu)$}

Given $\beta\in]0,\infty]$ we set $F_{\beta}(\mu):=-C(\mu,\nu_{P})+H(\mu)/\beta$ 
\begin{prop}
The gradient flow $\mu_{t}$ of $F_{\beta}$ on $\mathcal{P}_{2}(\R^{n})$
emanating from a given $\mu_{0}$ exists for any $\beta\in]0,\infty].$
Moreover, for $\beta<\infty$ we have that $\mu_{t}=\rho_{t}(x)dx$
where $\rho_{t}$ has finite Boltzmann entropy and Fisher information
and $\rho(x,t):=\rho_{t}(x)$ satisfies the following equation in
the sense of distributions on $\R^{n}\times]0,\infty[$ 
\begin{equation}
\frac{d\rho_{t}}{dt}=\frac{1}{\beta}\Delta\rho_{t}+\nabla\cdot(\rho_{t}\nabla\phi_{t}),\label{eq:evolution eq for perman limit in prop}
\end{equation}
where $\phi_{t}$ is the unique normalized solution in $\mathcal{C}_{P}$
to the equation \ref{eq:ma eq in text} and $\nabla\phi_{t}$ defines
a vector field with coefficients in $L_{loc}^{\infty}.$ \end{prop}
\begin{proof}
Given the previous lemma this follows immediately from Thm 8.3.1 and
Cor 11.1.8 in \cite{a-g-s} (the case $\beta=\infty$ has previously
been considered by Brenier \cite{br1b,Br2,br3} by lifting the problem
to the space of $L^{2}-$maps from $\R^{n}$ to $\R^{n}$ where Hilbert
space techniques can be applied ).
\end{proof}
More generally, as explained in the introduction it is natural to
introduce a parameter $\gamma\in[0,1]$ and a back-ground potential
$V(x),$ i.e. a convex function on $\R^{n}$ satisfying the growth
condition \ref{eq:growth of pot}. Then one replaces $E^{(N)}$ with
its weighted generalization $E_{\gamma,V}^{(N)}$ defined by formula
\ref{eq:def of weighted perm energy}. Then the previous proposition
still holds with $F_{\beta}$ replaced by the corresponding functional
$F_{\gamma,V}$ and $\phi_{t}$ in the evolution equation \ref{eq:evolution eq for perman limit in prop}
is replaced by $\gamma\phi_{t}+(1-\gamma)V$ and with $\beta=1$ (up
to rescaling time $t$ and the potential $V$ this is equivalent to
taking $\beta=\gamma)$

\subsection{The dynamic setting: Proof of Theorem \ref{thm:toric dynamic intro}}

By Proposition \ref{prop:main assump in perman case} the Main Assumptions
are satisfied.

\subsection{The static setting: Proof of Theorem \ref{thm:static toric}}

Let us first verify the properness assumption in Theorem \ref{thm:static intro}
holds for $\gamma<R_{P}.$ To this end we first assume that $b_{P}=0$
and observe that 
\begin{equation}
E^{(N)}(x_{1},....,x_{N})\geq-o(1)\sum_{i=1}^{N}|x_{i}|-o(1)\label{eq:proof of theorem static toric}
\end{equation}
Indeed, applying Jensen's inequality to the concave function log gives
\[
E^{(N_{k})}(x_{1},...,x_{N_{k}})=\frac{1}{k}\log\frac{1}{N_{k}!}\sum_{\sigma\in S_{N_{k}}}e^{k(x_{1}\cdot p_{\sigma(1)}+\cdots+x_{N}\cdot p_{\sigma(N_{k})})}-\frac{1}{k}\log\frac{1}{N_{k}!}\geq
\]
\[
\geq\frac{1}{N_{k}!}\sum_{j=1}^{N}x_{i}\cdot(\sum_{\sigma}p_{\sigma(i)})
\]
But for any fixed $i$ we have that $\sum_{\sigma}p_{\sigma(i)}=(N-1)!\sum_{p_{j}\in P\cap\Z/k}p_{j}$
and hence we get a Riemann sum: 
\begin{equation}
\frac{1}{N_{k}!}(\sum_{\sigma}p_{\sigma(i)})=\frac{1}{N_{k}}\sum_{p_{j}\in P\cap\Z/k}p_{j}:=b_{P}^{(k)}=b_{P}+o(1),\label{eq:def of k barycenter}
\end{equation}
 where $b_{P}:=\int_{P}p\nu_{P},$ which is assumed to vanish and
hence the inequality \ref{eq:proof of theorem static toric} follows.
But then the properness for $E_{\gamma,V}^{(N)}$ (defined by formula
\ref{eq:def of weighted perm energy}) in the case $b_{p}=0$ follows
immediately from the definition of $E_{\gamma,V}^{(N)}$ and the growth
assumption \ref{eq:growth of pot} on $V$ ensuring that that $V(x)\geq|x|/C-C$
since $0$ is assumed to be an interior point of $P.$ Finally, the
case then $b_{p}\neq0$ can be reduced to the previous case by translating
$P.$ More precisely, by the previous argument
\[
E_{\gamma,V}^{(N)}+o(1)\geq\sum_{i=1}^{N}\left(\gamma+o(1))x_{i}\cdot b_{P}+(1-\gamma)V(x_{i})\right)
\]
But, as shown in the proof of Theorem 2.18 in \cite{b-b} $R_{P}$
(defined by formula \ref{eq:inv R of conv bod}) is the sup of all
$r\in[0,1]$ such that $rx\cdot b_{P}+(1-r)\phi_{P}(x)\geq0$ and
since $|\phi_{P}(x)-V(x)|\leq C$ and $\gamma<R_{P}$ this gives the
desired properness. The convergence of the Boltzmann-Gibbs measures,
as $N\rightarrow\infty,$ now follows from Theorem \ref{thm:static intro}.
Moreover, the convergence as $t\rightarrow\infty$ for $N$ fixed
follows from well-known results about the linear Fokker-Planck equation
with a convex potential $E$ such that $\int e^{-E}dx<\infty$ (see
for example \cite{bgg} and reference therein). 

Finally, to prove the divergence of the partition function for $\gamma=R_{P}$
we first recall that if $\psi$ is a convex function on $\R^{d}$
then a necessary condition for $\int e^{-\psi}dx<\infty$ is that
$\psi(x)\rightarrow\infty$ as $|x|\rightarrow\infty$(as follows,
for example from Borell's lemma \cite[Lemma 3.1]{bo} which gives
that the integrability also holds for $\psi-\epsilon|x|,$ for $\epsilon$
any sufficiently small number \cite[Theorem 3.1]{bo}, and hence,
by Jensen's inequality, $\psi(x)-\epsilon|x|\geq-C_{\epsilon}:=-\log\int e^{-(\psi-\epsilon|x|)}dx).$
To violate the previous condition it is clearly enough to find a vector
$\vec{a}\in\R^{d}$ such that $t\mapsto\psi(t\vec{a})$ is an affine
function on $\R.$ We now consider the convex function $\psi:=E_{R_{P},\phi_{P}}^{(N_{k})}$
on $\R^{nN_{k}}$ and observe that for any fixed $a\in\R^{n}$ we
can write, with $\vec{a}:=(a,a,....,a),$ 
\begin{equation}
\psi(\vec{a})-\psi(0)=\left((1-R_{P})\phi_{P}(a)+R_{P}a\cdot b_{P}\right)+a\cdot\delta_{k},\,\,\,\delta_{k}:=R_{P}(b_{P}^{(k)}-b_{P})\in\R^{n}\label{eq:proof of toric static psi}
\end{equation}
 (compare formula \ref{eq:def of k barycenter}). Moreover, as shown
in the proof of Theorem 2.18 in \cite{b-b} when $a$ is taken as
a normal vector to a facet of $P$ containing the point $q$ appearing
in formula \ref{eq:inv R of conv bod}) the bracket in formula \ref{eq:proof of toric static psi}
vanishes. But this means that $t\mapsto\psi(t\vec{a})$ is an affine
function on $\R$ and hence $\int e^{-\psi}dx=\infty.$ Since $|\phi_{P}(x)-V(x)|\leq C$
the divergence also holds when $\phi_{P}$ is replaced by $V,$ which
concludes the proof.

\subsection{\label{sub:The-tropical-limit}The tropical setting }

The results above are also valid when the permanental interaction
energy $E^{(N_{k})}(x_{1},...,x_{N_{k}})$ is replaced by its tropical
analog, i.e. the following convex piecewise affine convex function 

\[
E_{trop}^{(N_{k})}(x_{1},...,x_{N_{k}}):=\max\sum_{\sigma\in S_{N_{k}}}x_{1}\cdot p_{\sigma(1)}+\cdots+x_{N}\cdot p_{\sigma(N_{k})}
\]
In other words this is a \emph{tropical permanent,} i.e. the permanent
of the rank $N$ matrix $(x_{i}\cdot p_{j})$ in the tropical semi-ring
over $\R$ (i.e. the set $\R\cup\{-\infty\}$ where the plus and multiplication
operations are defined by $\max\{a,b\}$ and $a+b,$ respectively
\cite{i-m}). Equivalently, in terms of discrete transport theory
this means that
\[
E_{trop}^{(N_{k})}(x_{1},...,x_{N_{k}}):=-C((\delta_{N}(x),\delta_{N}(p)))
\]
Passing to the tropical setting has, in particular, computational
advantages. Indeed, while all known methods for evaluating (general)
permanents take exponential time, the tropical permanent above is,
by its very definition, the optimal value of a linear assignment problem
and can be computed using an algorithm of cubic time complexity (see
the discussion in \cite{b-f-h-l-m}).

\section{Outlook}

In this final section we point out some relations between the limiting
evolution equation appearing in Theorem \ref{thm:toric dynamic intro}
(whose static solutions correspond to toric Kähler-Einstein metrics)
and other well-known evolution equations. We also indicate some relations
to sticky particle systems appearing at the microscopic level (i.e.
for finite $N)$ and the complex geometric picture. These relations
will be elaborated on in a sequel to the present paper \cite{b-o}.

\subsection{\label{sub:Relation-to-other evol eq}Relation to other evolution
equations and traveling waves}

In the one-dimensional case when $P:=[-a_{-},-a_{+}]$ integrating
the evolution equation for $\rho_{t}$ in Theorem \ref{thm:toric dynamic intro}
once reveals that the bounded decreasing function $u(x,t):=-\partial_{x}\phi_{t}$
(physically playing the role of a velocity field) satisfies\emph{
Burger's equation} \cite{ho} with positive \emph{viscosity }$\kappa:=\beta^{-1}:$
\[
\partial_{t}u=\kappa\partial_{x}^{2}u-u\partial_{x}u
\]
with the left and right space asymptotics $\lim_{x\rightarrow\pm\infty}u(x,t)=a_{\pm}.$
We recall that Burger's equation is the prototype of a non-linear
wave equation and a scalar conservation law (which is used, among
many other things, as a toy model for turbulence in the Navier-Stokes
equations \cite{f-b}). Interestingly, the barycenter $b_{P}$ of
the polytope $P$ coincides, in this one-dimensional situation, with
minus the \emph{speed }$s:=(a_{+}+a_{-})/2$ of the time-dependent
solution $u$ in the terminology of scalar conservation laws \cite{la}.
Hence, the vanishing condition $b_{P}=0,$ which in general is tantamount
to the existence of a stationary solution $\rho_{t}=\rho$ (as discussed
in connection to Theorem \ref{thm:toric dynamic intro}) simply means,
from the point of view of non-linear wave theory, that the speed $s$
vanishes. 

Similarly, the function $\phi(x,t):=\phi_{t}(x),$ which in complex
geometric terms is a Kähler potential, satisfies (after the appropriate
normalization) the following viscous Hamilton-Jacobi equation (known
as the deterministic KPZ-equation in the literature on growth of random
surfaces \cite{kpz,ha}): 
\begin{equation}
\partial_{t}\phi=\kappa\partial_{x}^{2}\phi+\frac{1}{2}(\partial_{x}\phi)^{2}.\label{eq:kpz eq}
\end{equation}
In the general higher dimensional case the evolution equation \ref{eq:evolut eq for toric ke}
(which is different than the higher dimensional version of Burger's
equation) can be seen as a dissipative viscous/diffusive version of
the semi-geostrophic equation appearing in dynamic meteorology (see
\cite{l,a-c-d-f} and references therein and \cite{Br2}  for a similar
situation in cosmology). Moreover, since 
\[
E(\mu)=-\frac{1}{2}d^{2}(\mu,\nu_{P})+\frac{1}{2}\int|x|^{2}\mu+C,
\]
where $d$ denotes the Wasserstein $L^{2}-$distance, the evolution
equation \ref{eq:evolut eq for toric ke} can also be seen as a quadratic
perturbation (with diffusion) of the ``geodesic flow'' on the Wasserstein
$L^{2}-$space (compare \cite[Example 11.2.10]{a-g-s}), which in
the one dimensional case appears in connection to the Sticky Particle
System \cite{ns}. As will be shown in \cite{b-o} the large time
asymptotics of the fully non-linear evolution equation \ref{eq:evolut eq for toric ke}
for the probability density $\rho_{t}$ in $\R^{n}$ are governed
by \emph{traveling wave solutions} in $\R^{n}$ whose speed coincide
with minus the barycenter $b_{P}$ of the convex body $P:$ 
\[
\rho_{t}(x)=\rho(x-b_{P}t)+o(t),\,\,\,\,\,\,t\rightarrow\infty
\]
 where the error terms $o(t)$ tends to zero in $L^{1}(\R^{n}$) (and
even in relative entropy) and where the limiting profile $\rho$ is
uniquely determined from a variant of the Monge-Ampère equation \ref{eq:k-e eq intro}
together with the condition that its barycenter coincides with the
barycenter of the initial data (thus breaking the translation symmetry).
In complex geometric terms $\rho$ corresponds to a certain canonical
Kähler-Einstein metric $\omega$ on $X$ with conical singularities
``at infinity'', playing the role of Calabi's extremal metrics in
this context. More generally, as will be elaborated on in \cite{b-o},
the results above apply in a more general setting where the measure
$\nu_{P}$ is multiplied by a density $g,$ which amount to replacing
the Monge-Ampère equation $MA(\phi)$ with $g(\nabla\phi)MA(\phi)$
and which from the point of view of scalar conservation laws corresponds
to a general concave flux function $f$ (when $n=1).$

\subsection{\label{sub:Outlook-on-the}The microscopic picture: sticky particles
in $\R^{n}$}

It can be shown that the attractive Newtonian interaction energy in
$\R$ is the one-dimensional version of the tropical permanental energy
$E_{trop}^{(N)}(x_{1},...x_{N})$ appearing in Section \ref{sub:The-tropical-limit}.
In the general higher dimensional setting it turns out that a very
concrete interpretation of the corresponding EVI gradient flow of
$E_{trop}^{(N)}(x_{1},...x_{N})$ on $\R^{nN}$ can be given; in particular
the particles perform zigzag paths with velocity vectors contained
in the polytope $-P$ generalizing the sticky $N-$particle system
on the real line. \footnote{When $n=1$ the dynamics is determined by the property that total
mass and momentum is conserved in collisions and that the particles
stick together when they collide; see \cite{b-g} and references therein.}Moreover, there is a static solution to the corresponding deterministic
$N-$particle system if and only if the ``discrete'' barycenter
of $P$ vanishes: 
\[
\frac{\text{1}}{N}(p_{1}+...+p_{N})=0
\]
which is consistent with the fact that the discrete barycenter can
be interpreted as the mean velocity of the particles. In general,
any initial configuration of points $(x_{1},...x_{N})(0)$ is assembled,
in a finite time, into a single particle $x_{*},$ namely the barycenter
of $\{x_{1},...x_{N}\}$ which moves at the mean velocity above. The
results in the present paper can also be used to study the large $N-$limit
of this deterministic system (which can be seen as a dissipative version
of the Hamiltonian particle system introduced in \cite{c-g-p} as
a discretization of the semigeostrophic equations). But the key point
of our approach is that it allows noise to be added to the particle
system. Then the role of the large $N-$limit of $x_{*}$ is played
by the volume form $\mu_{*}$ of a Kähler-Einstein metric on the toric
variety determined by the polytope $P$ (compare the discussion in
Section \ref{sub:Relation-to-other evol eq}). 

Interestingly, a similar particle system on $\R^{n}$ appears in Brenier's
approach \cite{Br2,br3} to the early universe reconstruction problem
in cosmology \cite{f-b} (in connection to the so called Zeldovich
approximation). In fact, our results can be used to validate the formal
large $N-$limit of the $N-$particle system with noise introduced
in \cite[Section 2.3]{br3}.\footnote{ As pointed out in \cite[Section 2.3]{br3} the formal argument used
there, which is based on the classical Freidlin-Wentzel theory (as
in \cite{da-g}), would require a Lipschitz bound on the drift (while
only a one-sided bound is available).}Details will appear in \cite{b-o}.

\subsection{\label{sub:The-complex-geometric}The complex geometric picture}

In this final section we provide some complex geometric motivation
for the present paper - a more detailed account, including the relations
to the Yau-Tian-Donaldson conjecture, will appear elsewhere

Let $X$ be an $n-$dimensional compact complex manifold. A metric
$g$ on $X$ is said to be Kähler-Einstein if $g$ has constant Ricci
curvature and $g$ is Kähler, i.e. in local holomorphic coordinates
$g$ can be represented as the real part of the positive definite
complex Hessian $\frac{\partial\phi(z)}{\partial z_{i}\partial\bar{z_{j}}}$
of a local function $\phi(z)$ called the Kähler potential of $g.$
If such a metric $g$ exists with positive Ricci curvature, then $X$
is necessarily a projective algebraic variety which is Fano, i.e.
the holomorphic (anti-canonical) line bundle $L:=\det(TX)$ over $X$
is positive. 

As shown in \cite{berm2} a Fano manifold comes with a sequence of
canonical $N-$particle random point process. The number of particles
$N$ arise as the pluri-antigenera of $X:$ 
\[
N=N_{k}:=\dim H^{0}(X,L^{\otimes k}),\,\,\,k=1,2,3,....
\]
where $H^{0}(X,L^{\otimes k})$ denotes the complex vector space consisting
of the global holomorphic sections of the $k$ th tensor power of
$L.$ The Fano condition ensures that $N_{k}\rightarrow\infty$ as
$k\rightarrow\infty.$ The local density of the corresponding canonical
symmetric probability measure $\mu^{(N_{k})}$ on $X^{N_{k}}$ is
defined by 
\begin{equation}
\rho^{(N_{k})}(z_{1},....,z_{N_{k}}):=\frac{1}{Z_{N_{k}}}\left|\det(z_{1},...z_{N_{k}})\right|^{-2/k},\,\,\,\,\,\,\det(z_{1},...z_{N_{k}}):=\det(s_{i}(z_{j})),\label{eq:prob density in neg curv}
\end{equation}
 where $\det(z_{1},...z_{N_{k}})\in H^{0}(X^{N_{k}},L^{\otimes k})$
is the Vandermonde type determinant formed from a given base $s_{1},...s_{N_{k}}$
in $H^{0}(X,L^{\otimes k})$ and $Z_{N_{k}}$ is the corresponding
normalization constant ensuring that the probability measure has unit
mass (by homogenity $\rho^{(N_{k})}$ is independent on the choice
of base). However, since the local density $\rho^{(N_{k})}(z_{1},....,z_{N_{k}})$
has singularities (for example when two points on $X$ merge) the
normalization constant $Z_{N_{k}}$ may be infinite, which means that
the random point processes are only well-defined if $Z_{N_{k}}<\infty.$
Such a Fano manifold $X$ was called Gibbs stable in \cite{berm2},
where it was shown that the condition can be rephrased in purely algebro-geometric
terms (see also \cite{fu} for further developments). It was conjectured
in \cite{berm2} that this condition is equivalent to $X$ admitting
a (unique) Kähler-Einstein metric (which necessarily has positive
Ricci curvature) whose volume form may be recovered as the deterministic
large $N-$limit of the empirical measures of the corresponding random
point processes. \footnote{The convergence of the processes in the opposite case when the dual
$\det(T^{*}X)$ of $\det(TX)$ is positive was settled in \cite{berm2}
(the limit is then the volume form of the unique Kähler-Einstein metric
on $X$ with negative Ricci curvature, whose existence was first established
in the seminal works of Aubin and Yau). }

The motivation for the present paper comes from a dynamic approach
to the latter conjecture where one introduces the interaction energy
\[
E^{(N_{k})}(z_{1},...,z_{N}):=\frac{1}{k}\log|\det(z_{1},...z_{N_{k}})|^{2}
\]
which is attractive, in the sense that it tends to $-\infty$ as two
particles merge. Locally, this object is represented by a plurisubharmonic
function, but in order to get a globally well-defined function on
$X^{N_{k}}$ one also has to fix a back-ground Kähler metric $g$
on $X$ (representing the first Chern class of $X)$ whose volume
form $dV_{g}$ then induces a metric $\left\Vert \cdot\right\Vert $
on $L$ which is used to replace the absolute values above. The point
is that, if $X$ is Gibbs stable, the canonical probability measure
$\mu^{(N_{k})}$ on $X^{N_{k}}$ can then be represented globally
as the corresponding Gibbs measure at inverse temperature $\beta=1$
(which is independent of the choice of metric $g):$ 
\[
\mu^{(N_{k})}=\frac{1}{Z_{N_{k}}}e^{-E^{(N_{k})}}dV_{g}^{\otimes N_{k}}\left(=\frac{1}{Z_{N_{k}}}\left\Vert \det\right\Vert ^{-2/k}dV_{g}^{\otimes N_{k}}\right)
\]
i.e. as a determinantal point process on $X$ at negative temperature.
The different zero-temperature case was studied in \cite{b-b-w}. 

At any rate, even if $X$ is not Gibbs stable one can still look at
the stochastic gradient flow of $E^{(N_{k})}$ on the $N_{k}-$fold
product of the Riemannian manifold $(X,g).$ From this dynamic perspective
Gibbs stability simply means that the corresponding stochastic process
has an invariant measure, to wit, $\mu^{(N_{k})}.$ Accordingly, the
natural dynamic generalization of the conjecture referred to above
is that a (unique ) Kähler-Einstein metric $g_{KE}$ exists precisely
when the stochastic gradient flow of $E^{(N_{k})}$ admits a stationary
measure and then its volume form $dV_{g_{KE}}$ can be recovered from
the joint large $N$ and large $t$ limit of the flow. More precisely,
conjecturally the large $N-$limit of the corresponding stochastic
gradient flows is described by the complex version of the evolution
equation \ref{eq:evolution eq for perman limit in prop}, obtained
by replacing the real Monge-Ampère operator with its complex counterpart.
The latter flow is, at least formally, the Wasserstein gradient flow
of a free energy type functional $F(\mu)$ on $\mathcal{P}_{2}(X,g)$
and $F$ can be identified with the K-energy functional on the space
of Kähler metrics in $c_{1}(X)$ (using the Calabi-Yau isomorphism)
\cite{berm4}. Unfortunately, the study of the latter flows is plagued
by various analytical difficulties stemming from the singularities
of $E^{(N_{k})}$ and the lack of convexity. For example, even in
the simplest case when $X$ is the Riemann sphere, i.e. the one-point
compactification of the complex plane $\C,$ so that $E^{(N_{k})}$
is simply the attractive logarithmic pair interaction between $N_{k}$
equal charges on $\C,$ the convergence of the large $N-$limit, for
a fixed time, is a long-standing open problem (however, see \cite{fou}
for very recent partial results).

\subsubsection{\label{sub:The-toric-setting}The toric setting and its tropicalization}

The complex geometric setting which is relevant to the present paper
appears when $X$ is a toric Fano manifold, i.e. $X$ admits a holomorphic
action of the the real $n-$torus $T$ such that $(X,T)$ can be realized
as an equivariant compactification of the complex torus $\C^{*n}$
(with is standard $T-$action) \cite{do}. Such a compactification
$X$ is determined by a convex polytope $P,$ which has the property
that under the dense embedding of $\C^{*n}$ into $X$ the complex
vector space $H^{0}(X,L^{\otimes k})$ may be identified with the
space of all holomorphic Laurent polynomials $f(z)$ on $\C^{*n}$
of the form 
\[
f(z)=\sum_{m\in kP\cap\Z^{n}}a_{m}z^{m}
\]
(using multindex notation). In particular, introducing an ordering
$m_{1},...m_{N_{k}}$ on the integer points of $kP\cap\Z^{n}$ gives
a bases $s_{m_{1}}(z),...,s_{m_{N_{k}}}$ of multinomials in $H^{0}(X,L^{\otimes k}),$
which can be used to represent 
\begin{equation}
\det(z_{1},...z_{N_{k}})=\sum_{\sigma\in S_{N}}(-1)^{\mbox{sign}(\sigma)}z_{1}^{m_{\sigma(1)}}\cdots z_{N_{k}}^{m_{\sigma(N)}}\label{eq:det in toric}
\end{equation}
Now, the real vector space $\R^{n}$ makes it appearance when introducing
logarithmic coordinates on $\C^{*n},$ i.e. as the image of the $\mbox{Log }$
map 
\[
\mbox{Log}:\,\,\C^{*n}\rightarrow\R^{n},\,\,\,z\mapsto x:=(\log|z_{1}|^{2},...,\log|z_{n}|^{2}),
\]
whose fibers are the orbits of the action action of $T.$ Using this
map $T-$invariant metrics on $L\rightarrow X$ with positive curvature
may be identified with convex functions $\phi(x)$ on $\R^{n}$ such
that $(\partial\phi)(\R^{n})\subset P.$ In this picture the permanental
density $\mbox{Per}\ensuremath{(x_{1},...,x_{N_{k}})}$ arises as
the push-forward to $\R^{n},$ under the Log map, of the determinant
density \ref{eq:det in toric}. In other words, the smooth convex
permanental energy $E_{per}^{(N)}(x_{1},...x_{N})$ (formula \ref{eq:energy as permanent intro})
on $\R^{n}$ is an averaged version of the singular plurisubharmonic
interaction energy $E^{(N_{k})}$ on $\C^{*n}:$ 
\begin{equation}
E_{per}^{(N)}(x_{1},...x_{N})=\frac{1}{k}\log\int_{T^{N_{k}}}e^{kE^{(N_{k})}}d\theta^{\otimes N_{k}}\label{eq:perm energy in terms of E}
\end{equation}
Similarly, its tropical version $E_{trop}^{(N)}(x_{1},...x_{N})$
is the piecewise affine convex function on $\R^{nN}$ obtained as
the tropicalization of the Laurent polynomial $\det(z_{1},...z_{N_{k}})$
on $\C^{*nN_{k}}$ \footnote{Incidentally, tropicalization may be interpreted as a zero-temperature
limit by writing the tropical sum $\max\{a,b\}$ as the limit of $T^{-1}\log(e^{\frac{1}{T}a}+e^{\frac{1}{T}b})$
as $T\rightarrow0;$ compare the discussion in \cite{i-m}. }. Accordingly, Theorem \ref{thm:toric dynamic intro} and Theorem
\ref{thm:static toric} should be seen in the light of the well-known
philosophy of replacing an elusive complex geometric problem by a
more tractable convex geometric one, by the process of tropicalization
(see, for example, \cite{i-m}). 

From the complex geometric point of view the equation \ref{eq:monge-ampere with gamma}
is precisely the one appearing in Aubin's continuity method for Kähler-Einstein
metrics. In the complex setting the sup over all $\gamma\in[0,1]$
for which the corresponding complex Monge-Ampère equation has a solution
coincides with the the differential geometric invariant $R(X)$ of
$X$ defined as the greatest lower bound on the Ricci curvature of
any Kähler metric $\omega\in c_{1}(X)$ \cite{sz}. Moreover, as first
shown in \cite{li1} the invariant $R_{P}$ of $P$ coincides with
$R(X)$ when $X$ is the toric Fano variety corresponding to $P.$
In this complex geometric context the last point in Theorem \ref{thm:static toric}
implies one of the inequalities in the conjectural equality relating
$R(X)$ to the algebro-geometric invariant $\gamma(X)$ introduced
in \cite{berm2}, when $X$ is a toric variety. To explain this we
recall that

\begin{equation}
\gamma(X):=\liminf_{k\rightarrow\infty}\mbox{lct}(\mathcal{D}_{k}/k)>1,\label{eq:invariant gamma}
\end{equation}
where $\mbox{lct}(\mathcal{D}_{k}/k)$ is the log canonical thresholds
of the divisor $\mathcal{D}_{k}$ in $X^{N_{k}}$ cut out by the holomorphic
section of $-K_{X^{N_{k}}}$ corresponding to $\det(z_{1},...z_{N_{k}}).$
The number $\mbox{lct}(\mathcal{D}_{k}/k)$ may be analytically defined
as the sup over all positive numbers $\gamma$ such that $\int_{X^{N_{k}}}\left\Vert \det\right\Vert ^{-2\gamma/k}dV^{\otimes N_{k}}$
is finite.
\begin{cor}
\label{cor:ineq on toric var}Let $X$ be a toric Fano variety. For
any positive integer $k$ 
\[
\frac{1}{k}\mbox{lct}(\mathcal{D}_{k})<R(X)
\]
In particular, on any toric Fano variety $X$ the following inequality
holds: \textup{
\[
R(X)\geq\min\{\gamma(X),1\}
\]
}\end{cor}
\begin{proof}
Write the complex coordinates $z^{\alpha}$ on $\C^{*n}$ as $z^{\alpha}=e^{2x^{\alpha}+2iy^{\alpha}},$
where $x^{\alpha}\in\R$ and $y^{\alpha}\in\R/[0,2\pi]\Z$ for $\alpha=1,..,n.$
In particular, $x=\log z$ in the notation of Section \ref{sub:The-toric-setting}.
Representing $dV=\frac{1}{(2\pi)^{n}}e^{-V}dx\wedge dy$ for a $T-$invariant
function $V$ on $\C^{*n},$ which we identify with a function $V(x)$
on $\R^{n}$ such that $(\partial\phi)(\R^{n})\subset P,$ gives
\[
\int_{\C^{*nN_{k}}}\left\Vert \det(z_{1},...z_{N_{k}})\right\Vert ^{-\gamma2/k}dV^{\otimes N_{k}}=
\]
\[
=\int_{x\in\R^{nN_{k}}}\left(\int_{y\in[0,2\pi[^{nN_{k}}}\left(|\det(z_{1},...,z_{N_{k}})|^{2}\right){}^{-\frac{\gamma}{k}}(\frac{1}{(2\pi)^{n}}dy){}^{\otimes N_{k}}\right)e^{(\gamma-1)V(x)}dx{}^{\otimes N_{k}},
\]
where $\left\Vert \cdot\right\Vert $ denotes the metric on $-K_{X^{N_{k}}}$induced
by $dV.$ Hence, applying Jensen's inequality to the convex function
$f(s)=s^{-\gamma/k}$ on $]0,\infty[$ and using the relation \ref{eq:perm energy in terms of E}
gives 
\[
\int_{\C^{*nN_{k}}}\left\Vert \det(z_{1},...z_{N_{k}})\right\Vert ^{-\gamma2/k}dV^{\otimes N_{k}}\geq\int\mbox{Per \ensuremath{(x_{1},...,x_{N_{k}})^{-\frac{\gamma}{k}}}}e^{(\gamma-1)V(x)}dx{}^{\otimes N_{k}}
\]
But, according to Theorem \ref{thm:static toric}, the latter integral
diverges if $\gamma\geq R_{P}$ and that concludes the proof. 
\end{proof}
The last statement in the previous corollary can be viewed as a toric
complement to Fujita's very recent result \cite{fu}, which in particular
gives that if $\gamma(X)\geq1$ on a Fano variety $X,$ then $X$
is K-semistable and hence $R(X)\geq1$ (which in the toric case is
equivalent to the existence of a Kähler-Einstein metric \cite{w-z}).
This gives an illustration of how the process of tropicalization can
be used to gain information on the original complex geometric situation.

\section{\label{sub:Intermezzo:-Otto's-formal}Appendix: The Otto calculus}

In this appendix we briefly recall Otto's \cite{ot} beautiful (formal)
Riemannian interpretation of the Wasserstein $L^{2}-$metric $d_{2}$
on $\mathcal{P}^{2}(\R^{n}).$ The material is included with the non-expert
in mind as a motivation for the material on gradient flows on $\mathcal{P}^{2}(\R^{n})$
recalled in Section \ref{sub:Gradient-flows-on}.

\subsection{The Otto metric}

For simplicity we will consider probability measures of the form $\mu=\rho dx$
where $\rho$ is smooth positive everywhere (in order to make the
arguments below rigorous one should also specify the rate of decay
of $\rho$ at $\infty$ in $\R^{n}).$ The corresponding subspace
of probability measures in $\mathcal{P}_{2}(\R^{n})$ will be denoted
by $\mathcal{P}.$ First recall that the ordinary\emph{ }``affine
tangent vector'' of a curve $\rho_{t}$ in $\mathcal{P}$ at $\rho:=\rho_{0},$
when $\rho_{t}$ is viewed as a curve in the affine space $L^{1}(\R^{n}))$
is the function $\dot{\rho}$ on $\R^{n}$ defined by 
\[
\dot{\rho}(x):=\frac{d\rho_{t}(x)}{dt}_{|t=0}
\]
Next, let us show how to identify $\dot{\rho}$ with a vector field
$v_{\dot{\rho}}$ in $L^{2}(\rho dx,\R^{n}),$ which, by definition,
is the (non-affine) ``tangent vector'' of $\rho_{t}$ at $\rho,$
i.e. $v_{\dot{\rho}}\in T_{\rho}\mathcal{P}.$ First, since the total
mass of $\rho_{t}$ is preserved we have $\int\dot{\rho}dx=0$ and
hence there is a vector field $v$ on $\R^{n}$ solving the following
continuity equation: 
\begin{equation}
\dot{\rho}=-\nabla\cdot(\rho v)\label{eq:rho dot in terms of V}
\end{equation}
 In geometric terms this means that
\begin{equation}
\rho_{t}dx=(F_{t}^{V})_{*}(\rho_{0}dx)+o(t),\label{eq:inf push forward relation}
\end{equation}
 where $F_{t}^{V}$ is the family of maps defined by the flow of $V.$
Now, under suitable regularity assumptions $v_{\dot{\rho}}$ may be
defined as the ``optimal'' vector field $v$ solving the previous
equation, in the sense that it minimizes the $L^{2}-$norm in $L^{2}(\rho dx,\R^{n}).$
The Otto metric is then defined by 
\begin{equation}
g(v_{\dot{\rho}},v_{\dot{\rho}})=\inf_{V}\int\rho|v|^{2}dx=\int\rho|v_{\dot{\rho}}|^{2}dx,\label{eq:def of riemannian norm}
\end{equation}
 which can be seen as the linearized version of the defining formula
\ref{eq:def of wasser} for the Wasserstein $L^{2}-$metric. By duality
the optimal vector field $v_{\dot{\rho}}$ may be written as $v_{\dot{\rho}}=\nabla\phi,$
for a unique normalized function $\phi$ on $\R^{n}$ (under suitable
assumptions).

\subsection{The microscopic point of view}

Let us remark that a simple heuristic ``microscopic'' derivation
of the Otto metric can be given using the isometry defined by the
empirical measure $\delta_{N}$ (Lemma \ref{lem:isometries}). Indeed,
given a curve $(x_{1}(t),...,x_{N}(t))$ in the Riemannian product
$(X^{N},\frac{1}{N}g^{\otimes N})$ with tangent vector $(\frac{dx_{1}(t)}{dt},...,\frac{dx_{1}(t)}{dt})$
at $t=0$ we can write its squared Riemannian norm at $(x_{1}(0),...,x_{N}(0))$
as 
\begin{equation}
\left\Vert (\frac{dx_{1}(t)}{dt},...,\frac{dx_{1}(t)}{dt})\right\Vert ^{2}=\int|v|^{2}\delta_{N}(0)\label{eq:microscopic otto}
\end{equation}
where $\delta_{N}(t):=\frac{1}{N}\sum\delta_{x_{i}(t)}$ and $v$
is any vector field on $X=\R^{n}$ such that $v(x_{i})=\frac{dx_{i}(t)}{dt}_{|t=0}.$
Note that setting $\rho_{t}:=\delta_{N}(t)$ the vector field $v$
satisfies the push-forward relation \ref{eq:inf push forward relation}
(with vanishing error term). Moreover, since passing to the quotient
$X^{N}/S_{N}$ does not effect the corresponding curve $\rho_{t}$,
minimizing with respect to the action of the permutation group $S_{N}$
in formula \ref{eq:microscopic otto} corresponds to the infimum defining
the Otto metric in formula  \ref{eq:def of riemannian norm}.

\subsection{Relation to gradient flows and drift-diffusion equations}

If $G$ is a smooth functional on $\mathcal{P}$ then a direct computations
reveals that its (formal) gradient wrt the Otto metric at $\rho$
corresponds to the vector field $v(x)=\nabla_{x}(\frac{\partial G(\rho)}{\partial\rho}).$
In other words, the gradient flow of $G(\rho)$ may be written as
\begin{equation}
\frac{\partial\rho_{t}(x)}{\partial t}=\nabla_{x}\cdot(\rho v_{t}(x)),\,\,\,\,\,\,v_{t}(x)=\nabla_{x}\frac{\partial G(\rho)}{\partial\rho}_{|\rho=\rho_{t}}\label{eq:formal gradient flow}
\end{equation}
In particular, for the Boltzmann entropy $H(\rho)$ (formula \ref{eq:def of H and I})
one gets, since $\frac{\partial G(\rho)}{\partial\rho}=\log\rho$
(using that the mass is preserved), that the corresponding gradient
flow is the heat (diffusion) equation and the gradient flow structure
then implies that $H(\rho_{t})$ is decreasing along the heat equation.
Moreover, a direct calculation reveals that $H$ is \emph{convex}
on \emph{$\mathcal{P}$} in sense that the Hessian of $H$ is non-negative
and hence it also follows from general principles that the squared
Riemannian norm $|\nabla H|^{2}(\rho_{t})$ is decreasing. In fact,
by definition $|\nabla H|^{2}(\rho)$ coincides with the Fisher information
functional $I(\rho)$ (formula \ref{eq:def of H and I}). More generally,
the gradient flow of the Gibbs free energy $F_{\beta}^{V}$ is given
by the diffusion equation with linear drift $\nabla_{x}V:$ 
\begin{equation}
\frac{\partial\rho_{t}}{\partial t}=\frac{1}{\beta}\Delta_{x}\rho_{t}+\nabla_{x}\cdot(\rho_{t}\nabla_{x}V),\label{eq:linear drift diff}
\end{equation}
 often called the linear Fokker-Planck equation in the mathematical
physics literature. The study of the previous flow using a variational
discretization scheme on $\mathcal{P}^{2}(\R^{n})$ was introduced
in \cite{j-k-o} (compare Section \ref{sub:Intermezzo:-Otto's-formal}).

\end{document}